\documentclass[acmsmall,screen,format=manuscript,nonacm]{acmart}

\AtBeginDocument{%
  }

\setcopyright{acmcopyright}
\copyrightyear{2018}
\acmYear{2018}
\acmDOI{XXXXXXX.XXXXXXX}

\acmConference[Conference acronym 'XX]{Make sure to enter the correct
  conference title from your rights confirmation email}{June 03--05,
  2018}{Woodstock, NY}
\acmPrice{15.00}
\acmISBN{978-1-4503-XXXX-X/18/06}

\usepackage{graphicx}
\usepackage{subfigure} 
\usepackage{diagbox}
\usepackage{amsthm}
\usepackage{xcolor}
\usepackage{multirow}
\usepackage{comment}
\usepackage{soul,color}
\usepackage{cancel}
\usepackage{tablefootnote}
\usepackage[colorinlistoftodos,prependcaption,textsize=small]{todonotes}
\usepackage[ruled, linesnumbered]{algorithm2e}
\usepackage{soul, color, xcolor}
\newtheorem{thm}{Theorem}
\newtheorem{lem}{Lemma}
\newtheorem{Definition}{Definition}
\newtheorem{Property}{Property}
\usepackage{svg}

\DeclareRobustCommand{\hlgreen}[1]{{\sethlcolor{green}\hl{#1}}}

\begin{document}

\title
{\textit{CrossCert}: A Cross-Checking Detection Approach to Patch Robustness Certification for Deep Learning Models}

\author{Qilin Zhou}
\affiliation{%
  \institution{City University of Hong Kong}
  \city{Hong Kong}
  \country{China}}
\email{qilin.zhou@my.cityu.edu.hk}

\author{Zhengyuan Wei}
\affiliation{%
  \institution{City University of Hong Kong \& The University of Hong Kong}
  \city{Hong Kong}
  \country{China}}
\email{zywei4-c@my.cityu.edu.hk}

\author{Haipeng Wang}
\affiliation{%
  \institution{City University of Hong Kong}
  \city{Hong Kong}
  \country{China}}
\email{haipeng.wang@my.cityu.edu.hk}

\author{Bo Jiang}
\affiliation{%
  \institution{
  Beihang University}
  \city{Beijing}
  \country{China}}
\email{jiangbo@buaa.edu.cn}

\author{W.K. Chan}
\affiliation{%
  \institution{%
  City University of Hong Kong}
  \city{Hong Kong}
  \country{China}}
\email{wkchan@cityu.edu.hk}
\authornote{Corresponding author}







\begin{abstract}
Patch robustness certification is an emerging kind of defense technique against adversarial patch attacks with provable guarantees. 
There are two research lines:
certified recovery and certified detection. 
They aim to label malicious samples with provable guarantees correctly
and 
issue warnings for malicious samples predicted to non-benign labels with provable guarantees, respectively.
However, existing certified detection defenders suffer from protecting labels subject to manipulation, and existing certified recovery defenders cannot systematically warn samples about their labels.
A certified defense that simultaneously offers robust labels and systematic warning protection against patch attacks is desirable.
This paper proposes a novel certified defense technique called \textit{CrossCert}. \textit{CrossCert} formulates a novel approach by cross-checking two certified recovery defenders to provide unwavering certification and detection certification. 
Unwavering certification ensures that a certified sample, when subjected to a patched perturbation, will always be returned with a benign label without triggering any warnings with a provable guarantee.
To our knowledge, \textit{CrossCert} is the first certified detection technique to offer this guarantee.
Our experiments show that, with a slightly lower performance than ViP and comparable performance with PatchCensor in terms of detection certification, \textit{CrossCert} certifies a significant proportion of samples with the guarantee of unwavering certification.

\end{abstract}




\authorsaddresses{%
This research is partly supported by the CityU MF\_EXT (project no. 9678180).\\ 
Authors’ addresses: Z. Zhou, H. Wang, W.-K. Chan, City University of Hong Kong, Hong Kong; emails: \{qilin.zhou, haipeng.wang\}@my.cityu.edu.hk, wkchan@cityu.edu.hk;
Z. Wei, City University of Hong Kong \& The University of Hong Kong, Hong Kong; emails: zywei4-c@my.cityu.edu.hk; B. Jiang, State Key Laboratory of CCSE, Beihang University, Beijing, China; emails: jiangbo@buaa.edu.cn}

\maketitle

\section{Introduction}
Deep learning (DL) models achieve high performance in many application tasks \cite{2021SurveyK,litjens2017survey}. 
However, they are vulnerable to adversarial attacks \cite{gu2019badnets,brown2017adversarial,zhong2021understanding,szegedy2013intriguing}. 
To address this weakness, defense techniques against adversarial attacks are under active research \cite{ma2019nic,usman2021nn,zhang2022towards,zhang2021out}.

Patch adversarial attack \cite{brown2017adversarial} is a model of physical adversarial attacks in the real world \cite{levine2020randomized}.
It manipulates pixels within a certain region of an image (called a patch region \cite{xiang2022patchcleanser}) to deceive DL models. 
For example, attaching a stop sign with a patch can raise a safety concern due to image misclassifications in autonomous driving scenarios \cite{eykholt2018robust}. 
A novel category of defenses, called patch robustness certification \cite{patchcensor},
is emerging and attracts the attention of the SE community \cite{patchcensor,zhou2023majority}.
It has two research lines:
certified recovery \cite{zhou2023majority,xiang2021patchguard,xiang2022patchcleanser,levine2020randomized,li2022vip,salman2022certified} and certified detection \cite{xiang2021patchguard++,han2021scalecert,li2022vip,patchcensor,mccoyd2020minority}.
Certified recovery aims to assign benign labels to malicious samples, and certified detection aims to issue warnings for malicious samples predicted to non-benign labels. They provide provable guarantees for benign samples on their goals, respectively known as recovery certification and detection certification.
%
Still, as we present in Section~\ref{sec:pre}, the former ones 
cannot systematically warn samples with non-benign labels, and the latter ones cannot
recover benign labels from manipulation.
%


This paper proposes \textit{CrossCert}, a novel cross-checking \textit{detection} approach to patch robustness certification. 
It addresses the challenge that existing certified detection defenders
\cite{xiang2021patchguard++,han2021scalecert,li2022vip,patchcensor,mccoyd2020minority}
 {cannot} recover a benign label for a malicious version of a certifiably detectable sample 
 (see Def. \ref{def:cert_detect}) while 
 remaining silent (i.e., issuing no warning). We formulate this notion as the \emph{unwavering certification} (see Def. \ref{def:cert_no_warning}).
\textit{CrossCert} provides provable guarantees for this notion (see Thm. \ref{thm:consistency})
and for detection certification (see Thm. \ref{thm: intersection}). 
  It composes a pair of certified recovery defenders (referred to as $R_1$ and $R_2$), outputs the prediction label returned by $R_1$, analyzes their individual recovery semantics, and cross-checks the results of the worst-case analysis.
  It aims to certify a given sample in the following aspects through the ideas of \emph{making no mistake} and \emph{making no common mistake}:
 %
  %
  Is the benign sample certifiably recoverable 
  (see Def. \ref{def:cert_recovery}) and predicted to the same label by both $R_1$ and $R_2$?
 %
 %
 Do all malicious versions of the benign sample that are manipulated within the same patch region fail to simultaneously attack $R_1$ and $R_2$ to produce the same malicious label?
  \textit{CrossCert} then refines the answers for these two questions to enforce the two guarantees mentioned above, which infer a pair of conditions to remain silent and issue warnings for malicious versions of those certifiably unwavering and certifiably detectable samples, respectively.
 

In our experiments on three benchmark datasets, including ImageNet\cite{deng2009imagenet}, CIFAR100 \cite{krizhevsky2009learning}, and CIFAR10 \cite{krizhevsky2009learning}, \textit{CrossCert} 
finds many certifiably unwavering samples.
For instance, it achieves up to 44.68\%, 52.50\%, and 78.75\% in certifiably unwavering accuracy with Vision
Transformer for these three datasets, respectively.
With the extra constraints of offering unwavering certification and labels from recovery defenders, it
achieves certifiably detectable accuracy that is slightly lower compared to the top-performing peer technique ViP \cite{li2022vip}, and is comparable to that of PatchCensor \cite{patchcensor} in certifiably detectable accuracy.



The main contribution of this paper is threefold. First, this paper is the first work to propose the notion of
unwavering certification. Second, it proposes a novel cross-checking framework \textit{CrossCert} with its theory for unwavering certification and detection certification. Last but not least, it presents an evaluation to show the feasibility and effectiveness of \textit{CrossCert}.

We organize the rest of this paper as follows. \S\ref{sec:pre} revisits the preliminaries.
\S\ref{sec:CC} and \S\ref{sec:eva} present \textit{CrossCert} followed by its evaluation. 
\S\ref{sec:rel} reviews closely related work. 
We conclude this work in \S\ref{sec:con}.

\section{Preliminaries}\label{sec:pre}

\subsection{Classification Model and Patch Attacker}
An image sample is a matrix of pixels with $h$ rows and $w$ columns.
Given an image sample space $\mathcal{X} \subset \mathbb{R}^{w \times h}$ 
and a label space $\mathcal{Y} = \{0, 1, \cdots, |\mathcal{Y}|-1\}$, we denote a classification model by a function $g: \mathcal{X} \rightarrow \mathcal{Y}$, which takes a sample $x \in \mathcal{X}$ as input to output a label $y \in \mathcal{Y}$. 

We use $\textsc{J}$ to denote {an all-ones matrix}
($\textsc{J} := \lambda i,j: \textsc{J}_{ij} =1$) and \textsc{O} to denote {an all-zeros matrix}
($\textsc{O} := \lambda i,j: \textsc{O}_{ij} =0$). 
Let $U$ and $V$ be binary matrices with respective elements  $U_{ij}$ and $V_{ij}$.
We define the elementwise matrix operators for addition, subtraction, and multiplication, $+$, $-$, and $\odot$, as follows:
(1) 
$U + V := \lambda i,j: \textit{max}\{U_{ij},V_{ij}\}$.
(2)
$U - V :=  \lambda i,j: U_{ij} - V_{ij}$ if $U_{ij}=1$, otherwise 0. 
(3) 
$U \odot V :=  \lambda i,j: U_{ij} \times V_{ij}$.


A patch attacker can modify pixels within a specific \textbf{patch region}
of an image, denoted by a matrix $\textsc{p} \in \mathbb{P} \subset [0,1]^{w \times h}$, where $\mathbb{P}$ represents the set of all possible patch regions, and all elements within and outside the patch region are set to 1 and 0, respectively. 
We express an attacker constraint set $\mathbb{A}(\textit{x})$ = $\{\textit{x}'
\mid \textit{x}'=(\textsc{J}-\textsc{p})\odot \textit{x}+\textsc{p}\odot \textit{x}'' \land \textsc{p}\in \mathbb{P} \}$ to represent the bound of the samples generated by a patch attacker, where $\textit{x} \in \mathcal{X}$ is the original image without modification (called a\textbf{ benign} sample), $\textit{x}' \in \mathcal{X}$ is an image provided by the attacker (called a\textbf{ malicious} sample), $\textit{x}'' \in \mathcal{X}$ is an arbitrary image.
We also say that the malicious sample $\textit{x}'$ {is around the benign sample} $\textit{x}$ if $\mathbb{A}(\textit{x})$ contains $\textit{x}'$.

Given a benign sample \textit{x}, the attacker aims to find a malicious sample $\textit{x}' \in \mathbb{A}(\textit{x})$ such that $g(\textit{x})\neq g(\textit{x}')=y_a$ for some $y_a \in \mathcal{Y}$. In that case,  
$\textit{x}'$ is called a \textbf{harmful} variant (or harmful sample) of $\textit{x}$ and $y_a$ is called a malicious label.
We further refer to the pair $(\textsc{p}, y_a)$ where $\textsc{p}$ is the patch region on $\textit{x}$ to constrain the attacker to generate
$\textit{x}'$ such that
$g(\textit{x}')=y_a$
as an \textbf{%
attack configuration}.



\subsection{Certified Defense against Patch Attacks: Detection and Recovery}
\textbf{Certified defense}  \cite{xiangshort} for
 detection  \cite{mccoyd2020minority,xiang2021patchguard++,patchcensor,han2021scalecert,li2022vip} 
and  recovery \cite{chiang2020certified,levine2020randomized,salman2022certified,li2022vip,chen2022towards,lin2021certified,xiang2021patchguard,metzen2021efficient,zhang2020clipped,xiang2022patchcleanser,zhou2023majority} 
aims to offer
\textbf{provable} guarantees on robustness against adversarial patch attacks.

\subsubsection{Certified Detection}

A certified detection defense aims to design a detection defender, denoted as
$D = \langle g, v, c_{\textsc{d}}\rangle$, which checks samples with the following behavior.

\begin{itemize}
\item 
The function $g(\textit{x})$ outputs a prediction label $g_\textit{x}$ for the input sample $\textit{x}$.
\item
The function $v(\textit{x})$ is called a warning verification function.
It returns \textit{True} if it detects the input sample \textit{x} harmful, otherwise \textit{False}.
So, a warning on \textit{x} is raised if $v(\textit{x})$ = \textit{True}.
\item 
The function $c_{\textsc{d}}(x)$ is called a detection certification function. 
It verifies whether $D$ \textit{guarantees} issuing warning for all 
harmful variants around the input sample \textit{x}.
If it can, \textit{x} is called \textbf{certifiably detectable} and 
 $c_{\textsc{d}}(x)$ returns \textit{True}; otherwise,  $c_{\textsc{d}}(x)$ returns \textit{False}.
\end{itemize}

The defender $D$ executes $c_{\textsc{d}}(.)$ to certify benign samples.
If the labels of a certifiably detectable sample and a malicious sample around it differ,
$v(.)$ should raise a warning when executing with the malicious sample as input.
Lastly, the function $g(.)$ returns prediction labels in all situations. 
 Def. \ref{def:cert_detect} captures the concept of \textbf{certified detection} ---
Suppose a certified detection defender certifiably detects a benign sample \textit{x}.
It must raise warnings for all harmful variants of \textit{x}.


\begin{Definition}[Certified Detection]\label{def:cert_detect}
\label{def:certification-definition}
A certified detection defender $D = \langle g, v, c_{\textsc{d}}\rangle$ certifies a benign sample $\textit{x}$ as certifiably detectable (i.e., $c_{\textsc{d}}(\textit{x})$ = \textit{True}) if
 $[\forall \textit{x}' \in \mathbb{A}(\textit{x})$, $g(\textit{x}')\neq g(\textit{x}) \implies$ $v(\textit{x}')=\textit{True}]$ holds.
\end{Definition}

%
If the defender issues a warning for a sample, it may further invoke a fallback strategy, such as discarding the sample or passing the sample and the warning outcome to the fallback tasks \cite{patchcensor}. 

\subsubsection{Certified Recovery} 
A certified recovery defense aims to design a recovery defender, denoted as
$R = \langle g, v, c_{\textsc{r}}\rangle$, which assigns a recovered label to an input sample
with the following behavior.

\begin{itemize}
\item
The function $g(\textit{x})$ outputs a prediction label $g_\textit{x}$ (called a \textbf{recovered label}) for the sample $\textit{x}$.

\item 
The function $v(\textit{x})$ is a warning verification function. 
It is not used by the existing recovery defenders. When absent, the function returns \textit{False} by default. We include this function in the formulation of a certified recovery defender to cover the scope of the present paper.

\item
The function $c_{\textsc{r}}(\textit{x})$ is a recovery certification function. 
It verifies whether 
$R$ \textit{guarantees} predicting the benign labels for all malicious versions of the benign sample $\textit{x}$.
If this is the case, $x$ is called \textbf{certifiably recoverable}, and $c_{\textsc{r}}(x)$ returns \textit{True}; otherwise, it 
returns \textit{False}.
\end{itemize}


%
Def. \ref{def:cert_recovery} states the concept of \textbf{certified recovery} --- 
If a certified recovery defender $R$ ensures all malicious versions of a benign sample \textit{x} predicted to the same label, 
it certifies \textit{x} as certifiably recoverable.
For certification, 
$R$ executes all three.
For prediction, it only executes $g(.)$ and $v(.)$.

\begin{Definition}[certified Recovery]\label{def:cert_recovery}
A certified recovery defender $R$ with 
$g$ and $c_{\textsc{r}}$ as its prediction and recovery certification functions 
certifies a benign sample $\textit{x}$ as certifiably recoverable (i.e., $c_{\textsc{r}}(\textit{x})$ = \textit{True}) if the condition $[\forall \textit{x}' \in \mathbb{A}(\textit{x})$, $g(\textit{x}')=g(\textit{x})]$ holds. 
\end{Definition}

\subsection{Voting-Based Recovery, and Masking-Based Detection and Recovery}\label{sec: V&M_detection_recovery}
Existing certified defenders  against adversarial patch attacks include 
voting-based recovery \cite{levine2020randomized,salman2022certified,li2022vip,chen2022towards,lin2021certified,xiang2021patchguard,metzen2021efficient,zhang2020clipped,zhou2023majority} as well as masking-based detection \cite{patchcensor,li2022vip,xiang2021patchguard++,han2021scalecert,mccoyd2020minority} and recovery \cite{xiang2022patchcleanser}. 

\subsubsection{Voting-Based Recovery}
\label{sec:voting-recovery}
The main idea of voting-based recovery defenders is to limit the capability of an attack patch on a sample to produce a malicious label.
There are two strategies in the literature.
Strategy 1 divides the sample into smaller pieces called ablations \cite{levine2020randomized,xiang2021patchguard, salman2022certified,li2022vip,zhou2023majority}, and Strategy 2 limits the size and impact of the feature extraction windows to be small and slides these windows through the sample \cite{zhang2020clipped,metzen2021efficient,xiang2021patchguard}. 
Both strategies generate a set of mutants for the sample.
The defender then determines the majority label among these mutants. 
By analyzing the worst-case scenario against all patch regions, 
if the majority label remains unchanged, the defender certifies the sample as certifiably recoverable.

In the rest of this section, we revisit the concept of voting-based recovery using Strategy 1 due to its simpler presentation. Note that defenders using Strategy 2 are similar to those using Strategy 1.

An \textbf{ablated region} is represented by a matrix $\textsc{b}\in\mathbb{B}\subset [0,1]^{w \times h}$, where $\mathbb{B}$ is a set that contains all ablated regions.
In $\textsc{b}$, the elements within and outside the region are set to 1 and 0, respectively. 
An \textbf{ablated mutant} $\textit{x}_\textsc{b}$ of a sample \textit{x} is generated by applying the ablated region \textsc{b} to the sample \textit{x} (i.e., $\textit{x}_\textsc{b}$=$\textit{x}\odot\textsc{b}$).
In the resulting mutant $\textit{x}_\textsc{b}$, only the overlapping part of \textit{x} and $\textsc{b}$ remains visible (see Fig. \ref{fig:exmaple}). 
We denote the set of all mutants of \textit{x} as $\mathbb{X}_\mathbb{B}(\textit{x})$. 

Suppose $f$ represents the base model in a voting-based recovery defender. 
The prediction label assigned by $f$ to a mutant 
is referred to as a \textit{vote}.
In the worst case, an attacker can modify any pixels within $\textsc{b}\odot\textsc{p}$ in some $\textit{x}_\textsc{b}$ with \textsc{b} overlapping with \textsc{p} (i.e., $\textsc{b}\odot\textsc{p}\neq\textsc{O}$) to cause $f$ to output a specific label.

A vote-based certified recovery defender $R_v = \langle g, v, c_\textsc{r} \rangle$ returns the majority of these votes as its recovered label (i.e., ${g(\textit{x})}=$
\emph{arg}$max_{y}
|\{\textit{x}_\textsc{b} \in \mathbb{X}_\mathbb{B}(\textit{x}) \mid f(\textit{x}_\textsc{b}) = y\}|=$ {$g_x$}). 
We use the symbol $\Delta$ to denote the maximum number of ablated mutants that any patch can overlap with. Note that this number is calculated based on the geometric relationship (i.e., sizes and shapes)  between the ablated region, the patch region, and the image \cite{levine2020randomized,salman2022certified}.
Previous works \cite{levine2020randomized,salman2022certified,li2022vip} have proven that for a benign sample $\textit{x}$, if the number of votes received by the majority label is greater than the number of votes received by every non-majority label by at least $2\Delta$ (i.e., 
$c_\textsc{r}(\textit{x})$ is defined as 
{$|\{\textit{x}_\textsc{b} \in \mathbb{X}_\mathbb{B}(\textit{x}) \mid f(\textit{x}_\textsc{b})=g_x\}|-max_{y\neq g_x}|\{\textit{x}_\textsc{b} \in \mathbb{X}_\mathbb{B}(\textit{x}) \mid f(\textit{x}_\textsc{b})=y\}| > 2\Delta$})
, then the condition $[\forall \textit{x}'\in\mathbb{A}(\textit{x})$,  ${g(\textit{x}')}=$\emph{arg}$max_{y}
|\{\textit{x}_\textsc{b} \in \mathbb{X}_\mathbb{B}(\textit{x}) \mid f(\textit{x}_\textsc{b}) = y\}|=g(\textit{x})]$ holds, which means that the sample $\textit{x}$ is certifiably recoverable. 
In Fig. \ref{fig:exmaple}, the diagram on the left illustrates an example. 
\begin{figure}[htb]
\centering
\includegraphics[width = .48\linewidth]{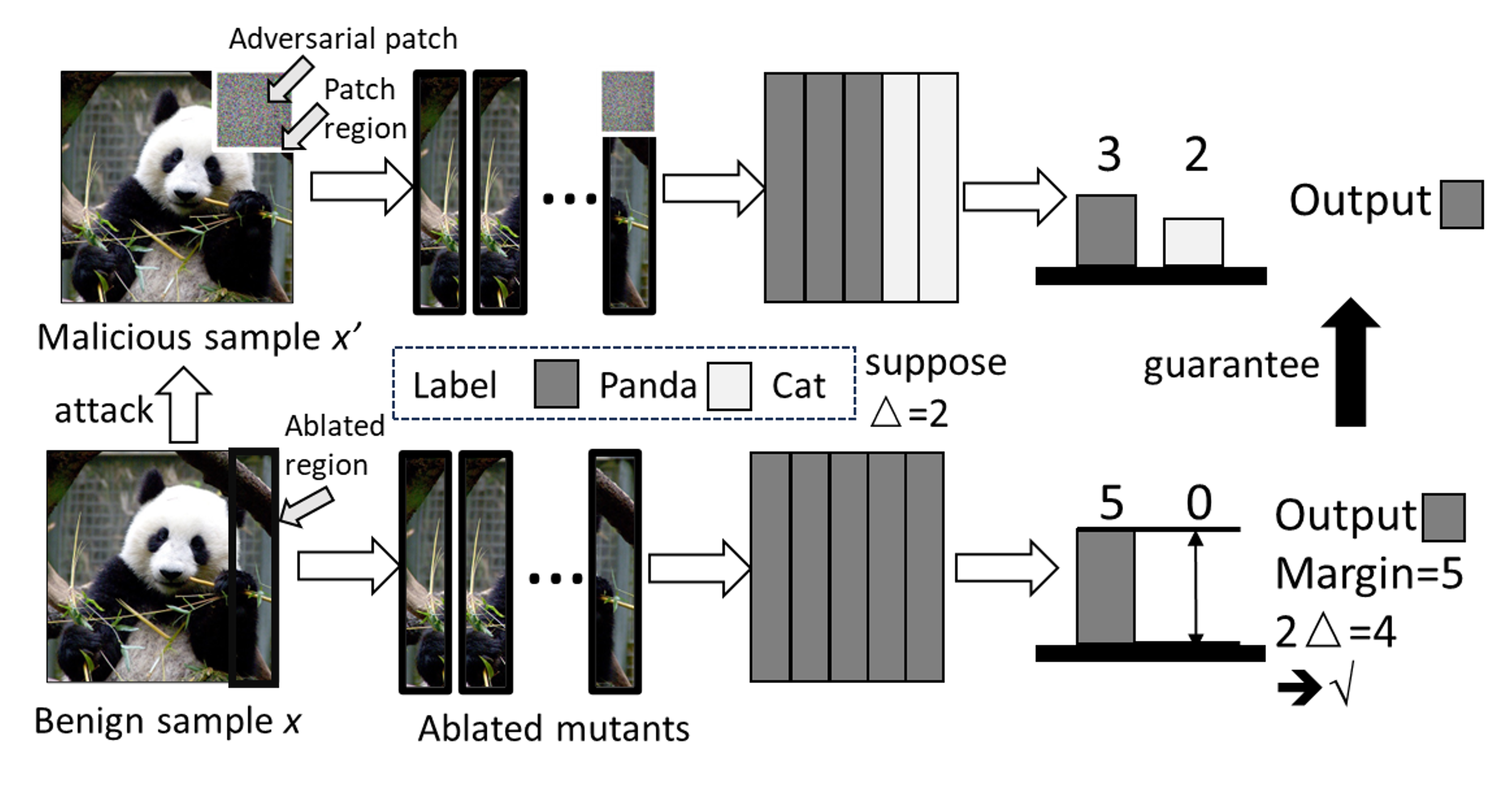}\hfill
\includegraphics[width = .48\linewidth]{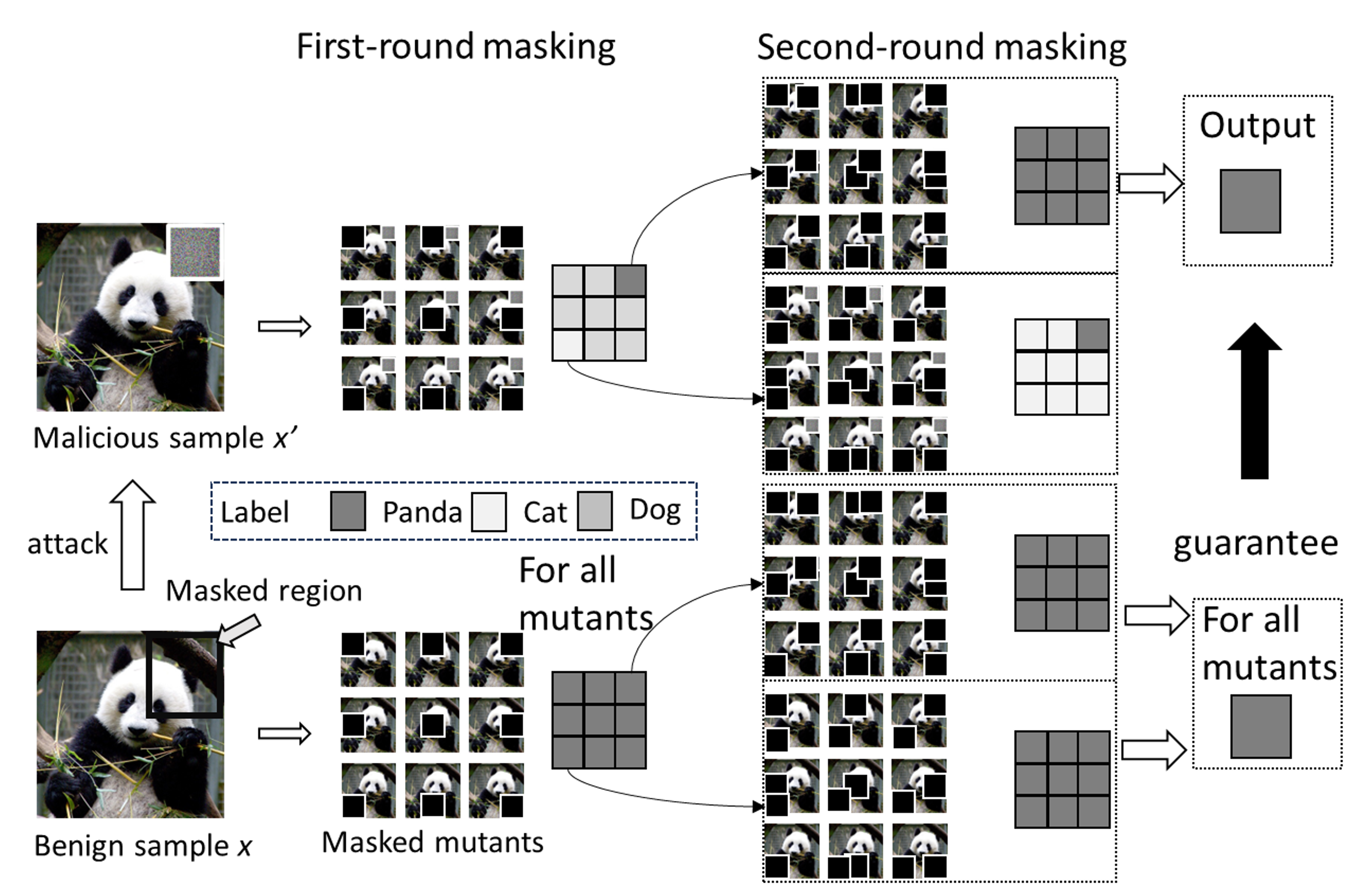}
\caption{Examples for a voting-based recovery defender (left) and a masking-based recovery defender (right). 
\textit{Left:} Given a benign sample \textit{x}, the voting-based recovery defender performs a round of ablating on \textit{x} to generate ablated mutants. 
Suppose $\Delta=2$. 
As illustrated in the lower section, for a benign sample, $f$ predicts Panda for 5 mutants and Cat for none. Thus, we have $5-0 > 2\Delta$, and the sample is predicted as Panda. Then, if there is any patch attached to this benign sample (as illustrated in the upper section, for example), at most two mutants would be affected to change the prediction from Panda to Cat in the worst case, resulting in 3 votes for Panda and 2 for Cat. Still, we have $3 > 2$, and the malicious sample is predicted as Panda. 
\textit{Right:} Given a benign sample $\textit{x}$, 
a masking-based recovery defender applies two rounds of masking on $\textit{x}$ to generate mutants. Suppose the prediction labels of all these mutants reach a consensus (all are Panda, as illustrated in the lower section). Then, if there is any patch (which must be masked by at least one mask) overlapping with $\textit{x}$, Alg.~\ref{alg:patchcleanser_original} guarantees to output Panda for any resulting malicious sample.
The upper section illustrates Case II in Alg.~\ref{alg:patchcleanser_original}. In the first-round masking, mutants may have different prediction labels like Panda, Cat, and Dog. 
Alg.~\ref{alg:patchcleanser_original} continues to perform second-round masking on first-round masked mutants with non-majority labels. 
In Alg.~\ref{alg:patchcleanser_original}, if all mutants of any first-round masked mutant reach a consensus on their predicted labels, then the label is output, which is the case for the first-round masked mutant in the top-right corner, where $f$ predicts Panda for all its second-round masked mutants. Thus, Panda is output.}
\label{fig:exmaple}
\end{figure}

\subsubsection{Masking-Based Detection and Recovery} 
\label{sec:masking-recovery}
The main idea of masking-based detection is as follows: 
A mask hides a specific region of a sample to generate a one-masked 
mutant (see Fig. \ref{fig:exmaple}).
After generating a set of masks that ensures each patch region is covered by at least one mask, masking-based certified detection defenders certify the sample as certifiably detectable if it verifies that the base model assigns the benign label of the benign sample to all its one-masked mutants.


A \textbf{mask region} (mask for short) is represented by a matrix $\textsc{m}\in\mathbb{M}\subset [0,1]^{w \times h}$, where $\mathbb{M}$ is a set of mask regions that can cover every patch region $\textsc{p}\in\mathbb{P}$ by at least one mask (i.e., $\forall \textsc{p}\in\mathbb{P}, \exists \textsc{m}\in\mathbb{M}, \textsc{p}\odot\textsc{m}=\textsc{p}$), and all elements within and outside each mask region are set to 1 and 0, respectively. 
A \textbf{(masked) mutant} $\textit{x}_\textsc{m}$ of a sample \textit{x} is generated by removing the elements in \textit{x} that overlap with any element equal to 1 in the mask region \textsc{m}
(i.e., $\textit{x}_\textsc{m}=(\textsc{J}-\textsc{m})\odot\textit{x}$ where $\textsc{J}$ is a matrix with all elements equal to 1).
The set of mutants generated from \textit{x} by a mask set $\mathbb{M}$  is denoted as $\mathbb{X}_\mathbb{M}(\textit{x})$.
Note that a mutant is also a sample.
We also simply write $(\textsc{J}-\textsc{m}-\textsc{m}')\odot\textit{x}$ as the shorthand notation of $(\textsc{J}-\textsc{m}) \odot (\textsc{J}-\textsc{m}')\odot\textit{x}$%
, which removes the elements in \textit{x} that overlap with any element equal to 1 in the mask region $\textsc{m}$ and $\textsc{m}'$.

%
%
%


Suppose $h$ is the base model used in a masking-based certified recovery defender. 
We also refer to the prediction label assigned by $h$ to a masked mutant as a vote.

Masking-based certified detection defenders 
use the prediction label assigned by $h$ to the input sample as the prediction label (i.e., ${g(\textit{x})}=h(\textit{x})$) \cite{patchcensor, li2022vip}. 
They then check against a consistency condition: 
whether all the votes
are the same as the prediction label assigned by $h$ to the input sample (i.e., $\forall \textit{x}_\textsc{m}\in \mathbb{X}_\mathbb{M}(\textit{x}), h(\textit{x}_\textsc{m})=h(\textit{x})$).
Existing works \cite{patchcensor, li2022vip} have proven that if a benign sample \textit{x} satisfies the consistency condition, they 
guarantee to detect any harmful variant of $\textit{x}$ by checking whether the variant violates the consistency condition.
Therefore, the function $c_\textsc{d}(.)$ certifies a benign sample as certifiably detectable if the sample satisfies the consistency condition
{and the function $v(.)$ raises a warning if the input sample violates the consistency condition}. 


\begin{figure*}[t]
\begin{minipage}[t]{0.463\linewidth}
\begin{algorithm}[H]
\scriptsize
\caption{{\small PatchCleanser's Prediction}}
\label{alg:patchcleanser_original}
\SetKwInOut{KwIn}{Input}
\SetKwInOut{KwOut}{Output}
\KwIn{$\textit{x}$ $\gets$ sample, \\
      $h$ $\gets$ base classification model, \\
      $\mathbb{M}$ $\gets$ mask set \\
      }
\KwOut{
$g_\textit{x}$ $\gets$ output label\\
}

\text{$y_{\textit{maj}}=\text{arg\,max}_{y}
|\{\textsc{m} \in \mathbb{M} \mid h(\textit{x}_\textsc{m}) = y \land \textit{x}_\textsc{m}=(\textsc{J}-\textsc{m})\odot\textit{x} \}|$}
\\
${\mathbb{M}_{\textit{min}}}=\{\textsc{m} \in \mathbb{M} \mid h(\textit{x}_\textsc{m}) \neq y_{\textit{maj}}\land \textit{x}_\textsc{m}=(\textsc{J}-\textsc{m})\odot\textit{x}\}$\\
\If{$\mathbb{M}_{\textit{min}}=\emptyset$}
{
\Return{$y_{\textit{maj}}$} \tcp*{\tiny Case I: agreed prediction}
}
\ForEach{$\textsc{m}$ $\in$ $\mathbb{M}_{\textit{min}}$\\} 
{
\text{\tiny
$y_{\textit{maj}}'=\text{arg\,max}_{y}
|\{\textsc{m}' \in \mathbb{M}\mid h((\textsc{J}-\textsc{m}-\textsc{m}')\odot\textit{x})=y \}|$}\\
${\mathbb{M}_{\textit{min}}}'=\{\textsc{m}' \in\mathbb{M} \mid h((\textsc{J}-\textsc{m}-\textsc{m}')\odot\textit{x})\neq y_{\textit{maj}}'\}$\\
\If{${\mathbb{M}_{\textit{min}}}'=\emptyset$}
{
\text{\Return{$y_{\textit{maj}}'$ \tcp*{\text{\tiny Case II: disagreed prediction}}}}
}
}
\Return{$y_{\textit{maj}}$} 
\tcp{\tiny Case III: majority prediction}

\end{algorithm}
\end{minipage}
\begin{minipage}[t]{0.525\linewidth}
\begin{algorithm}[H]
 \scriptsize
\caption{{\small  \text{PatchCleanser's Revised Prediction}}}
\label{alg:patchcleanser_revision}
\SetKwInOut{KwIn}{Input}
\SetKwInOut{KwOut}{Output}
\KwIn{$\textit{x}$ $\gets$ sample, \\
      $h$ $\gets$ base classification model, \\
      $\mathbb{M}$ $\gets$ mask set \\
      }
\KwOut{
$g_\textit{x}$ $\gets$ output label, 
$v_\textit{x}$ $\gets$ warning label\\
}

$y_{\textit{maj}}=\text{arg\,max}_{y}
|\{\textsc{m}\in \mathbb{M}\mid h(\textit{x}_\textsc{m}) = y \land \textit{x}_\textsc{m}=(\textsc{J}-\textsc{m})\odot\textit{x}\}|$\\
${\mathbb{M}_{\textit{min}}}=\{\textsc{m} \in \mathbb{M} \mid h(\textit{x}_\textsc{m}) \neq y_{\textit{maj}}\land \textit{x}_\textsc{m}=(\textsc{J}-\textsc{m})\odot\textit{x}\}$\\
\If{$\mathbb{M}_{\textit{min}}=\emptyset$}
{
\Return{$y_{\textit{maj}}$, \textit{False}} \tcp*{\tiny Case I: agreed prediction}
}

\ForEach{ {$\textsc{m}$ $\in$ $\mathbb{M}$}
\tcp{ \textbf{change from} $\textsc{m}\in\mathbb{M}_{\textit{min}}$}
} 
{
$y_{\textit{maj}}'=\text{arg\,max}_{y}
|\{\textsc{m}' \in \mathbb{M} \mid h((\textsc{J}-\textsc{m}-\textsc{m}')\odot\textit{x}) = y \}|$\\
${\mathbb{M}_{\textit{min}}}'=\{\textsc{m}' \in \mathbb{M} \mid h((\textsc{J}-\textsc{m}-\textsc{m}')\odot\textit{x})\neq y_{\textit{maj}}'\}$\\
\If{${\mathbb{M}_{\textit{min}}}'=\emptyset$}
{
\text{\Return{$y_{\textit{maj}}'$, \textit{False}\ \tcp*{\text{\tiny Case II: disagreed prediction}}}}
}
}

\Return{$y_{\textit{maj}}$, \textit{True}}  \tcp*{\text{\tiny Case III: majority prediction + \textbf{warning}}}
\end{algorithm}
\end{minipage}
\end{figure*}

However, such a detection defender {cannot determine} which label is benign. PatchCleanser \cite{xiang2022patchcleanser} is the first (and, to our knowledge, the most representative) masking-based recovery defender.
It formulates a double masking strategy to address this issue and provides a recovered label prediction algorithm, as shown in Alg. \ref{alg:patchcleanser_original}, to serve as the function $g(.)$.
Alg. \ref{alg:patchcleanser_original} consists of three cases:
\begin{itemize}
\item[Case I]
 Lines 1--5 show the procedure of masking-based certified detection. 
 If all mutants vote for the same label, PatchCleanser will output that label as the recovered label.
 \item[Case II]
In lines 6--9, PatchCleanser iterates through all mutants whose vote is different from the majority label.
It iteratively performs a second round of masking on these mutants (lines 8--9). 
In an iteration, if this masking operation on the mutant for the iteration does not result in any label disagreement among the resulting mutants of that mutant (line 10), the algorithm outputs the majority (sole) label among these resulting mutants (line 11). 
\item[Case III]
Otherwise,
PatchCleanser outputs the majority label of the mutants produced by the first round of masking as the recovered label (line 14).
\end{itemize}

 Xiang et al. \cite{xiang2022patchcleanser}
 have proven that if all mutants of a benign sample \textit{x} produced in both the first and second rounds of masking yield the same label (i.e., $\exists y \in \mathcal{Y}, \forall \textsc{m}_0, \textsc{m}_1 \in \mathbb{M}, h((\textsc{J}-\textsc{m}_0-\textsc{m}_1)\odot \textit{x})=y$, which we refer to this condition as the two masking agreement (\textbf{TMA}) condition for \textit{x}), 
  Alg. \ref{alg:patchcleanser_original} guarantees to output the benign label (the label $y$ in the TMA condition) for every malicious sample of \textit{x}.
As such,
a masking-based certified recovery defender $R_m = \langle g, v, c_\textsc{r} \rangle $ modeled after PatchCleanser is defined as: $g$ is Alg. \ref{alg:patchcleanser_original}$, $ $c_\textsc{r}(\textit{x})$ is the TMA condition for \textit{x}, and
$v(\textit{x})$ always returns \textit{False}. The diagram on the right in Fig. \ref{fig:exmaple} illustrates an example.

\subsubsection{Critical Limitation}
\label{sec:limitation-of-cd+cr}
Nonetheless, combining the existing designs of certified detection and certified recovery defenders presents challenges. 
In the following discussion, we discuss the masking-based recovery and detection defenders
because a direct comparison is enabled between them.

We first define the following concept. 
The mutant set $\mathbb{X}_\mathbb{M}(\textit{x}'')$ of a sample or a mutant, denoted by $\textit{x}''$, is called \emph{self-consistent} iff the base model $h$ of the defender assigns the same prediction label to all the mutants in the mutant set.
We recall that 
in a masking-based detection defender $D_m = \langle g, v, c_\textsc{d} \rangle$, 
$v(.)$ return \textit{False} if the mutation set of the input sample is self-consistent and all mutants vote to $g_x$; otherwise, return \textit{True} to raise a warning.
We also recall that 
in a masking-based recovery (PatchCleanser), Cases II and III in Alg.\ref{alg:patchcleanser_original} require the mutation set of the input sample is not self-consistent (or it would visit Case I instead). Note that Cases II and III are the main use cases to recover benign labels against patch attacks since the condition to output in Case I is fragile against attacks.

Here is the conflict: the detection defender \emph{always} raises a warning in the main use cases of the recovery defender, even if the recovered label is benign. 
In other words, there is \emph{no} safety indicator to protect the recovered label, such as to warn users that the sample could be malicious and the label is not safe to use, since ``warning'' will always be raised. 
(Indeed, we agree that this is a hard problem.)

Simply modifying the existing detection defender may make the situation even worse. For instance, adding any non-trivial condition to the warning verification function alone will \emph{destroy} the property of the certified detection guarantee of the detection defender. 

We are well aware that since the inception of either line of research, the certified accuracy (for detection and recovery) has gradually increased \cite{li2022vip, salman2022certified,saha2023revisiting} over the years. 
However, we are unaware of existing techniques combining these two classes of techniques with compatible semantics, 
such as producing a recovered label from certified recovery for certified detection.

\section{\textit{CrossCert}}\label{sec:CC}
\subsection{Overview}
\textit{CrossCert} belongs to a novel class of certified detection defenders.
It proposes a novel cross-examination framework, in which
it composes and cross-checks the label recovery semantics of its base certified recovery defenders.
Like previous detection defenders \cite{mccoyd2020minority,xiang2021patchguard++,han2021scalecert,li2022vip,patchcensor}, it gives a detection certification.
On top of this first level of certification, 
\textit{CrossCert} provides the next level of certification, called \emph{unwavering certification}:
It verifies whether it can guarantee that no warnings will be issued to any malicious versions of a given certifiably recoverable sample (with a benign recovered label), 
making \textit{CrossCert} a more comprehensive and practical certified defender against adversarial patch attacks.
(For example, a certified photo in a passport can be matched with the live image of that person at a custom self-service counter without issuing a warning, even if his/her face is partially injured or obscured.)
It outputs the recovered labels generated by its primary base certified recovery defender as its prediction labels. Therefore, its prediction label generation principle is intrinsically compatible with certified recovery.



\subsection{Unwavering Certification}
This section presents a new certification target --- \textbf{unwavering certification} --- as defined in Def. \ref{def:cert_no_warning}.

\begin{Definition}[Unwavering Certification]\label{def:cert_no_warning}
A certified detection defender $D = \langle g,v,c\rangle$ reports a benign sample $\textit{x}$ to be certifiably unwavering if $[\forall \textit{x}' \in \mathbb{A}(\textit{x})$, $g(\textit{x}')=g(\textit{x})\land v(\textit{x}')=\textit{False}]$ holds.
\end{Definition}

Def. \ref{def:cert_no_warning} states that 
a benign sample is called  \textbf{certifiably unwavering} for a certified detection defender if the defender can always assign the benign sample and all its malicious versions to the same label without issuing any warnings for these malicious samples.


%
If an existing certified recovery defender \cite{xiang2022patchcleanser,levine2020randomized,salman2022certified,metzen2021efficient,zhou2023majority,xiang2021patchguard,li2022vip} receives a malicious version of a benign sample that is not certifiably recoverable as input, it cannot distinguish this malicious version from those malicious samples that it produces benign recovered labels, making the defender challenging to alert their users.
On the other hand, existing certified detection defenders \cite{patchcensor,li2022vip,xiang2021patchguard++,mccoyd2020minority,han2021scalecert} are insensitive to the recovered labels generated by certified recovery defenders, leading to a lack of systematic warning protection for these recovered labels.

The notion of \textit{unwavering certification} is designed to address this lack of systematic warning protection problem. Since the main designs are a warning verification function and a new certification function, a defender for unwavering certification falls naturally into the class of certified detection. 
Property \ref{property-cu-imply-cr}
shows that unwavering certification is a novel subclass of detection certification.

\begin{Property}[] 
\label{property-cu-imply-cr}
If a sample is certifiably unwavering, it must be certifiably detectable.
\end{Property}
\begin{proof}
Recall that a certifiably detectable sample \textit{x} must satisfy the condition $[\forall \textit{x}' \in \mathbb{A}(\textit{x})$, $g(\textit{x}')\neq g(\textit{x}) \implies$ $v(\textit{x}')=\textit{True}]$. If $x$ is a certifiably unwavering sample, the antecedent $g(\textit{x}')\neq g(\textit{x})$ will never hold because Def. \ref{def:cert_no_warning} requires that the condition $g(\textit{x}') = g(\textit{x})$ holds for all malicious versions of \textit{x}. It implies that the condition $[\forall \textit{x}' \in \mathbb{A}(\textit{x})$, $g(\textit{x}')\neq g(\textit{x}) \implies$ $v(\textit{x}')=\textit{True}]$ holds.
\end{proof}

Since each certifiably unwavering sample is a certified detectable sample, increasing the proportion of certified detectable samples as certifiably unwavering decreases the probability of invoking a fall-back task (that takes over the situations when the input samples are warned),
resulting in 
a higher degree of automation (without the interruption of the fall-back task) with a safety guarantee.

\subsection{Design of the Base Framework: \textit{CrossCert-base}}\label{sec:cc-base}
This subsection presents the base framework of \textit{CrossCert}, which we refer to as \textit{CrossCert-base}.
The basic idea of \textit{CrossCert-base} is to cross-check the outputs from two certified recovery defenders and raise warnings when their outputs are inconsistent. The challenge addressed by \textit{CrossCert-base} is that it has to simultaneously ensure the unwavering certification and detection certification, which requires an in-depth analysis, although its solution is neat and tidy.

We first define $\mathbf{1}[z]$ as the indicator function to return \textit{True} if the condition $z$ holds, or else \textit{False}.

 
\textit{CrossCert-base} $D = \langle g, v, c\rangle $ 
 integrates two certified recovery defenders, denoted by $R_1=\langle g_1, v_1, c_1\rangle$ and $R_2=\langle g_2, v_2, c_2\rangle$. 
 Its components are formulated as follows:
 \begin{itemize} 
 \item 
 $g(\textit{x}):=g_1(\textit{x})$. 
 \item 
 $v(\textit{x}):=\mathbf{1}[g_1(\textit{x})\neq g_2(\textit{x})]$. 
 \item
 $c(\textit{x}):=\langle c_\textsc{u}(\textit{x}), c_\textsc{d}(\textit{x})\rangle $. 
 \item 
 $c_\textsc{u}(\textit{x}):=\mathbf{1}[g_1(\textit{x})=g_2(\textit{x})\land c_1(\textit{x})=c_2(\textit{x})=\textit{True}]$. 
 It
gives an unwavering certificate.
 \item
 $c_\textsc{d}(\textit{x}):=\mathbf{1}[[c_1(x)=\textit{True}]\lor[c_2(\textit{x})=\textit{True}\land g_1(\textit{x})= g_2(\textit{x})]]$.
 It gives a detection certificate.
 \end{itemize}


The design of \textit{CrossCert-base} is to directly output the recovered labels returned by the underlying (primary) base certified recovery defender $R_1$ in making predictions for all samples.
So,  \textit{CrossCert-base} is also a certified recovery defender by using $c_1$ of $R_1$ as its recovery certification function. As such, 
it shares the same certified recovery accuracy with $R_1$.
Furthermore,
a special case of \textit{CrossCert-base} is  $R_1 = R_2$. In this special case, 
 \textit{CrossCert-base} is degenerated into $R_1$ (a pure certified recovery defender) and cannot give any warning since the output labels of $R_1$ and $R_2$ are the same. 

We next present two theorems that support the two certifications ($c_\textsc{u}$ and $c_\textsc{d}$) stated above. 

Thm. \ref{thm:theory_1} is a direct consequence of the definition of certified recovery. The underlying intuition is that if both recovery defenders can recover the benign label of every malicious version $\textit{x}'$ of a benign sample \textit{x}, they must return the same label for every malicious version of \textit{x}.
Compared to Def. \ref{def:cert_no_warning}, if a defender can further ensure that it does not raise any warning on every such malicious version, the benign sample is a certifiably unwavering sample.
As stated above, the function $v(.)$ in \textit{CrossCert-base} is designed to raise a warning on an input sample only if the two base recovery defenders return different labels.
Thus, a benign sample \textit{x} that satisfies the antecedent of Thm. \ref{thm:theory_1}, which is the condition used to define the function $c_\textsc{u}(.)$ in \textit{CrossCert-base}, is certified by \textit{CrossCert-base} as certifiably unwavering (i.e., $c_\textsc{u}(\textit{x})$ = \textit{True}).
%
%
%
The proof for  Thm. \ref{thm:theory_1} 
only uses the outputs of recovery defenders for the logic deduction. It readily applies to all types of recovery defenders. In our experiments, to show the novelty of our framework,
we use a masking-based certified recovery defender as $R_1$ and a voting-based certified recovery defender as $R_2$.

\begin{thm}[A Condition for Unwavering Certification]
\label{thm:theory_1}
Given a sample \textit{x} and two recovery defenders $R_1=\langle g_1, v_1, c_1\rangle$ and $R_2=\langle g_2, v_2,  c_2\rangle$. If $g_1(\textit{x})=g_2(\textit{x})\land c_1(\textit{x})=c_2(\textit{x})=\textit{True}$, then the condition 
$[\forall \textit{x}'\in \mathbb{A}(\textit{x}), g_1(\textit{x}')=g_1(\textit{x})\land g_1(\textit{x}')=g_2(\textit{x}')]$ holds.
\end{thm}
\begin{proof}
Following the definition of certified recovery in Def. \ref{def:cert_recovery}, 
the condition $c_1(\textit{x})=c_2(\textit{x})=\textit{True}$ implies the condition $[\forall \textit{x}'\in\mathbb{A}(\textit{x}), g_1(\textit{x}')=g_1(\textit{x}),g_2(\textit{x}')=g_2(\textit{x})]$. Given that $g_1(\textit{x})=g_2(\textit{x})$, we have the condition $[\forall \textit{x}'\in\mathbb{A}(\textit{x}), g_1(\textit{x}')=g_2(\textit{x}')]$ holds.
\end{proof}




Thm. \ref{thm:new-cd-cert} establishes a novel connection between the recovery certification provided by the underlying recovery defenders and the guarantee of issuing warnings in the \textit{CrossCert-base} framework.
The underlying intuition is as follows:
If one of the two base recovery defenders can recover the benign label of every malicious version of a benign sample, then our defender can detect the presence of the mistakenly recovered label of every malicious version of that benign sample produced by the other defender in the pair.
As a result, we design our detection certification function
$c_\textsc{d}(\textit{x})$ to return \textit{True} if the condition $[c_1(x)=\textit{True}]\lor[c_2(\textit{x})=\textit{True}\land g_1(\textit{x})= g_2(\textit{x})]$ holds.
In this way,
by Thm. \ref{thm:new-cd-cert}, we obtain that the condition 
$[\forall \textit{x}'\in\mathbb{A}(\textit{x}), g_1(\textit{x}')= g_1(\textit{x})\lor g_2(\textit{x}')= g_2(\textit{x})=g_1(\textit{x})]$ should also hold
%
if $c_\textsc{d}(\textit{x}) = True$.
It implies that  the condition 
$[g_1(\textit{x}')=g_2(\textit{x}')]$ 
must be violated if the condition $[g_1(\textit{x}')\neq g_1(\textit{x})]$ is met. 
By the formulation of the warning verification function $v(.)$ of \textit{CrossCert-base},  $v(\textit{x}')$ must return \textit{False} for every malicious version $\textit{x}'$ of the benign sample \textit{x} under this condition.
By Def. \ref{def:cert_detect},
\textit{x} is a certifiably detectable sample.

\begin{thm}[A Condition for Detection Certification]
\label{thm:new-cd-cert}
Given a sample \textit{x} and two recovery defenders $R_1=\langle g_1, v_1, c_1\rangle$ and $R_2=\langle g_2, v_2,  c_2\rangle$. 
If the condition
$[c_1(x)=\textit{True}]\lor[c_2(\textit{x})=\textit{True}\land g_1(\textit{x})= g_2(\textit{x})]$
holds, 
then the condition $[\forall \textit{x}' \in \mathbb{A}(\textit{x})$, $g_1(\textit{x})\neq g_1(\textit{x}') \implies g_1(\textit{x}')\neq g_2(\textit{x}')]$ holds.
\end{thm}
\begin{proof}
Case 1: 
If the condition $[c_1(x)=\textit{True}]$ holds, then by Def. \ref{def:cert_recovery}, the condition $[\forall \textit{x}'\in\mathbb{A}(\textit{x}), g_1(\textit{x}')= g_1(\textit{x})]$ holds, which implies that the antecedent $[g_1(\textit{x}')\neq g_1(\textit{x})]$ would never hold. Case 2:
If the condition $[c_2(\textit{x})=\textit{True}\land g_1(\textit{x})= g_2(\textit{x})]$ holds, then, by Def. \ref{def:cert_recovery}, the condition $[\forall \textit{x}'\in\mathbb{A}(\textit{x}), g_2(\textit{x}')= g_2(\textit{x})=g_1(\textit{x})]$ holds. If $g_1(\textit{x}')\neq g_1(\textit{x})$, we must have $g_1(\textit{x}')\neq g_2(\textit{x}')$.
\end{proof} 
%
\subsection{\textit{CrossCert} Framework Design}\label{sec:cc_framework}
This section presents the design of \textit{CrossCert}. Fig. \ref{fig:Overview} depicts the overview of \textit{CrossCert}.

\begin{figure}[tb]
\centering
\includegraphics[width=0.99\textwidth]{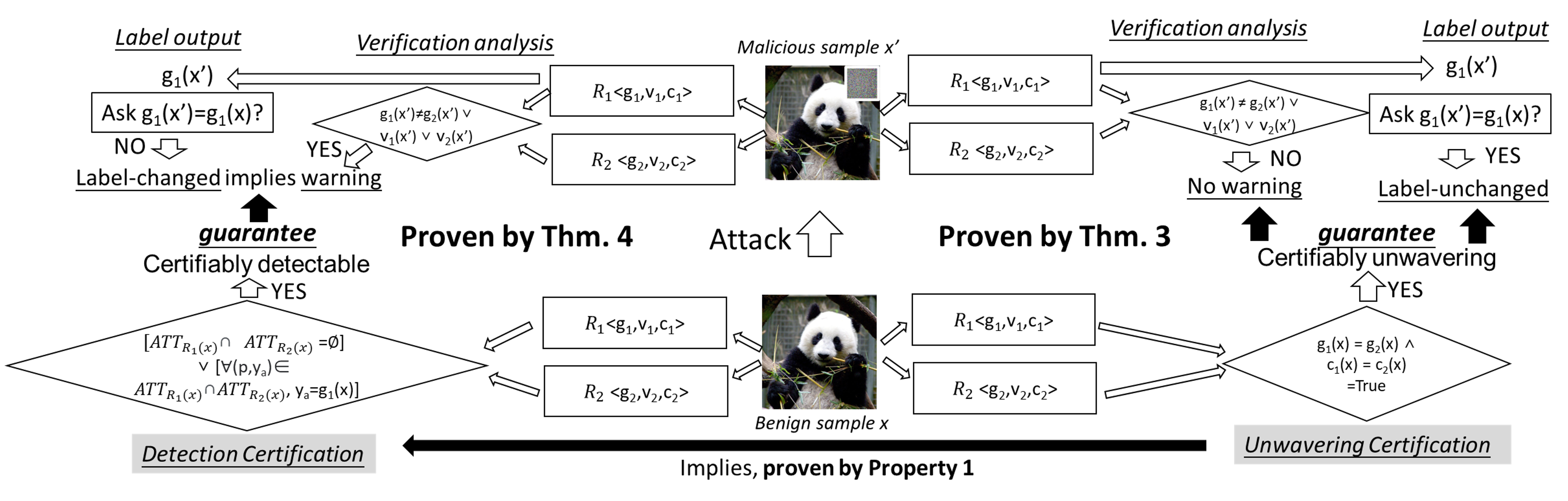}
\caption{Overview of \textit{CrossCert}. \textit{CrossCert} adopts a masking-based recovery defender 
$R_1 = \langle g_1, v_1, c_1\rangle$ and a voting-based recovery defender $R_2 = \langle g_2, v_2, c_2\rangle$. If the respective conditions are met, \textit{x} is certifiably detectable and certifiably unwavering, respectively. A certifiably unwavering sample is also a certifiably detectable sample, proven by Property \ref{property-cu-imply-cr}. For a certifiably detectable sample \textit{x}, \textit{CrossCert} guarantees that it must issue a warning for any malicious version $\textit{x}'$ around \textit{x} in the warning verification analysis if the recovered labels of \textit{x} and $\textit{x}'$ differ, proven by Thm. \ref{thm: intersection}. For a certifiably unwavering sample $\textit{x}$, \textit{CrossCert} guarantees that no label change can occur for any malicious samples around $\textit{x}$ and will not raise any warning in the warning verification analysis, proven by Thm. \ref{thm:consistency}.}
\label{fig:Overview}
\end{figure}

\subsubsection{Limitation of CrossCert-base}
The certification analyses of \textit{CrossCert-base} are only built on top of the composition of the underlying recovery defenders' outputs.
The underlying relationship (among the labels, mutants, and certification conditions of recovery defenders) to produce inconsistent outputs has yet to be understood.
%
Our \textbf{insight} is:
{if we can determine under what condition the pair of recovery defenders do not make the same mistake (e.g., do not output the same non-benign 
label for each malicious version of a benign sample), 
we can certify the sample as certifiably detectable.
%
To achieve this goal, we analyze under what conditions the base certified recovery defenders may mistakenly recover a non-benign label for a malicious version of the given benign sample. We present these analyses in Section \ref{sec:analysis-masking-recovery} for masking-based recovery and Section \ref{sec:analysis-voting-recovery} for voting-based recovery. We formulate \textit{CrossCert} by extending \textit{CrossCert-base} (see Section \ref{sec:formulation_of_cc}) to incorporate them along with the theoretical guarantee (see Section \ref{sec:guarantee_of_cc}).

\subsubsection{Formulation of \textit{CrossCert}}
\label{sec:formulation_of_cc}
\textit{CrossCert} $D = \langle g, v, c\rangle $ 
composes two certified recovery defenders, where the primary and secondary base recovery defenders are denoted by $R_1=\langle g_1, v_1, c_1\rangle$ and $R_2=\langle g_2, v_2, c_2\rangle$, respectively, where $R_1$ is a masking-based defender (i.e., $R_m$ described in Section \ref{sec:masking-recovery}) with a revised version of the prediction function $g_1$ (i.e., Alg. \ref{alg:patchcleanser_revision}), and $R_2$ is a voting-based defender (i.e., $R_v$ described in Section \ref{sec:voting-recovery}).
The components in \textit{CrossCert} are defined below:

 \begin{itemize}
     \item 
 $g(\textit{x}):=g_1(\textit{x})$.
 \item 
 $v(\textit{x}):=\mathbf{1}[g_1(\textit{x})\neq g_2(\textit{x})\lor v_1(\textit{x})=\textit{True}\lor v_2(\textit{x})=\textit{True}]$. 
 \item 
$c(\textit{x}):=\langle c_\textsc{u}(\textit{x}), c_\textsc{d}(\textit{x})\rangle$.
 \item
   $c_\textsc{u}(\textit{x}):=\mathbf{1}[g_1(\textit{x})=g_2(\textit{x})\land c_1(\textit{x})=c_2(\textit{x})=\textit{True}]$.   
\item 
 $c_\textsc{d}(\textit{x}):= \mathbf{1}[
 [\mathbb{ATT}_{R_1}(\textit{x})\cap\mathbb{ATT}_{R_2}(\textit{x})=\emptyset]\lor[\forall (\textsc{p},y_a)\in\mathbb{ATT}_{R_1}(\textit{x})\cap\mathbb{ATT}_{R_2}(\textit{x}), y_a=g_1(\textit{x})]]$. 
  \end{itemize}
  

Compared to \textit{CrossCert-base}, in \textit{CrossCert}, the warning verification function $v(.)$ covers more cases. 
Indeed, one design goal of \textit{CrossCert} is to guarantee detection certification, which requires theorem development because we need to ensure the covering of all possible malicious samples of every certifiably detectable, and this coverage cannot be verified/measured by empirical experiments, resulting in Thm.~\ref{thm: intersection}.
The result of Thm.~\ref{thm: intersection} is in the style of 
``$C_1 \implies [\forall \textit{x}' \in \mathbb{A}(\textit{x}), g(\textit{x}) \neq g(\textit{x}') \implies C_2]$''. The conditions $C_1$ and $C_2$ become the functions $c_\textsc{d}(.)$ and  $v(.)$ in \textit{CrossCert}, respectively.
In \textit{CrossCert-base}, the functions $v_1(.)$ and $v_2(.)$ always return \textit{False}, which 
is simplified 
into the current condition for $v(.)$ in \textit{CrossCert-base}.
Another design goal of \textit{CrossCert} is to guarantee unwavering certification, which also requires theorem development, resulting in Thm.~\ref{thm:consistency}.
Thm.~\ref{thm:consistency} is in the style of ``$C_3 \implies 
 [\forall \textit{x}' \in \mathbb{A}(\textit{x}), g(\textit{x}) = g(\textit{x}') \land ${$\neg C_2]$}''.
 The condition $C_3$ becomes the function  $c_\textsc{u}(.)$ in  \textit{CrossCert}, and \textit{CrossCert-base} also adopts this function. Here, $C_3$ deduces the negation of $C_2$, which corresponds to our purpose of suppressing warnings for the samples.
%
 Additionally, same as \textit{CrossCert-base},  \textit{CrossCert} is also a certified recovery defender.

In the following subsections, we will first present how we modify {$g_1(.)$} in \textit{CrossCert-base} (i.e., {replacing} Alg. \ref{alg:patchcleanser_original} by Alg. \ref{alg:patchcleanser_revision}) to support \textit{CrossCert}.
Then, we will present how we determine a set that contains all possible attack configurations of a recovery defender, each capturing where a patch \emph{can} be placed in a benign sample \textit{x} to produce malicious samples
and which malicious label for the resulting malicious samples derived from \textit{x} can be and yet the defender does not issue a warning on it (as Lem. \ref{lem:pc_analysis} for $R_1$ ($R_m$) and Lem. \ref{lem:voting_analysis} for $R_2$ ($R_v$)). 
After expanding the set with all possible pairs of a patch region and the benign label of \textit{x} (to include all cases in which attackers fail to change the prediction label), we obtain an expanded set and refer to it as $\mathbb{ATT}_{R_1}(\textit{x})$ for $R_1$ and $\mathbb{ATT}_{R_2}(\textit{x})$ for $R_2$ stated in the formulation of \textit{CrossCert} above.
Finally, we will formulate the conditions using these two expanded sets (based on the above-stated insight) to offer the detection certification guarantee (Thm. \ref{thm: intersection}) and present the recovery certification guarantee of Alg. \ref{alg:patchcleanser_revision} (Lem. \ref{lem:pc_analysis_revised}) and the unwavering certification guarantee of \textit{CrossCert} (Thm. \ref{thm:consistency}).



\subsubsection{Analysis of Masking-Based Recovery}
\label{sec:analysis-masking-recovery}
The warning verification function $v(.)$ of \textit{CrossCert} uses the warning verification functions of the two underlying recovery defenders for cross-checking. However, existing recovery defenders do not provide any warnings for their recovered labels, even if those are suspicious.
This section presents our minor modification on PatchCleanser (Alg. \ref{alg:patchcleanser_original}). 
We will show the correctness of the modified algorithm and 
derive a set to contain the possible malicious labels that the benign sample's malicious versions may attempt to mislead the defender, along with their patch regions.

Recall that PatchCleanser cannot guarantee the output in Case III in Alg. \ref{alg:patchcleanser_original} to be benign, but a recovered label that PatchCleanser guarantees to be benign may be output in Case III.
Therefore, we modify Alg. \ref{alg:patchcleanser_original} in such a way that no label warranted by PatchCleanser to be benign will be output in Case III. 
By doing so, we can associate a warning (which will be returned by $v_1(.)$) with every recovered label to indicate whether PatchCleanser \emph{fails to} guarantee this recovered label to be benign. 


Alg. \ref{alg:patchcleanser_revision} shows our lightly revised version of Alg. \ref{alg:patchcleanser_original}.
In the original algorithm (Alg.~\ref{alg:patchcleanser_original}), lines 6--13 enumerates the mutants in $\mathbb{M}_{\textit{min}}$ (i.e., these mutants producing non-majority labels).
A carefully crafted malicious version of a certifiably recoverable sample can lead the algorithm to generate a non-empty $\mathbb{M}_{\textit{min}}'$ at line 9, thereby outputting a benign label at line 14.
To fix this issue, we modify line 6 from iterating over $\mathbb{M}_{\textit{min}}$ to iterating over $\mathbb{M}$. 

We next present our analysis result based on Alg. \ref{alg:patchcleanser_revision}.
Lem. \ref{lem:pc_analysis_revised} proves the guarantee for the benign label recovery and the suppression of warning.
We need this guarantee because the \textit{CrossCert} framework is built on top of Alg. \ref{alg:patchcleanser_revision} to serve as its prediction function $g(.)$.
The proof follows what we have discussed above and the idea in the original proof for Alg. \ref{alg:patchcleanser_original} in \cite{xiang2022patchcleanser}.
Lem. \ref{lem:pc_analysis_revised} is a direct consequence of the TMA condition for our revised algorithm (Alg. \ref{alg:patchcleanser_revision}). 

\begin{lem}[Recovery Certification of PatchCleanser with Revised Prediction]
\label{lem:pc_analysis_revised}
Let  $R_m=\langle g, v, c \rangle$ be a masking-based recovery, 
in which the prediction function $g$ is Alg. \ref{alg:patchcleanser_revision}. 
If the TMA condition (i.e., $[\exists y \in \mathcal{Y}, \forall \textsc{m}_0, \textsc{m}_1 \in \mathbb{M}, h((\textsc{J}-\textsc{m}_0-\textsc{m}_1)\odot \textit{x})=y]$ where $h$ is the base classifier of $R_m$) holds for a benign sample \textit{x}, then {the label of every malicious version of \textit{x} returned by $g$ must be benign} (i.e., the condition 
$[\forall \textit{x}'\in\mathbb{A}(\textit{x}), g(\textit{x}')=g(\textit{x})]$ holds), and Alg. \ref{alg:patchcleanser_revision} will never reach line 14 (Case III).
\end{lem}

\begin{proof}
Compared to Alg. \ref{alg:patchcleanser_original}, Alg. \ref{alg:patchcleanser_revision} does not modify the logic for Case I (lines 1--5). 
We now prove how Alg. \ref{alg:patchcleanser_revision} recovers the benign label, excluding the cases handled by Case I. 
Suppose Alg. \ref{alg:patchcleanser_revision} accepts
a malicious sample as input such that the sample is produced by the attack configuration ($\textsc{p},y_a$).
Suppose the benign version \textit{x}  of the malicious sample satisfies the TMA condition (see Section~\ref{sec:masking-recovery}), which is $[\exists y \in \mathcal{Y}, \forall \textsc{m}_1, \textsc{m}_2 \in \mathbb{M}, h((\textsc{J}-\textsc{m}_1-\textsc{m}_2)\odot \textit{x})=y]$. Note that $g(\textit{x})=y$, output by Case I.
Then, by the definition of the mask set, for any patch region $\textsc{p}\in\mathbb{P}$, there must exist a mask $\textsc{m}_0\in \mathbb{M}$ such that $\textsc{m}_0\odot\textsc{p}=\textsc{p}$. Therefore, the mutant produced from the malicious sample $\textit{x}'$ with the mask $\textsc{m}_0$ must be the same as the mutant from the benign sample $\textit{x}$ with the mask $\textsc{m}_0$. 
(We arbitrarily choose $\textsc{m}_0 = \textsc{m}_2$ instead of
$\textsc{m}_0 = \textsc{m}_1$.)
So, we obtain $[
 \forall \textsc{m}_1 \in \mathbb{M},
h((\textsc{J}-\textsc{m}_1-\textsc{m}_0)\odot \textit{x}')=y]$ since $[
 \forall \textsc{m}_1 \in \mathbb{M},
h((\textsc{J}-\textsc{m}_1-\textsc{m}_0)\odot \textit{x})=y]$ holds due to the TMA condition being held on \textit{x}.
As lines 6-9 enumerate all masks, 
an empty $\mathbb{M}_{\textit{min}}'$ must be found in line 9 when enumerated to $m_0$.
Therefore, we must have $y_{\textit{maj}}' = y = g(x)$ output at line 11, and line 14 is never reached.
%
\end{proof}

Suppose \textit{x} is a benign sample that satisfies the TMA condition.
As presented in Section \ref{sec:masking-recovery}, 
PatchCleanser with Alg. \ref{alg:patchcleanser_original} for label prediction will certify 
\textit{x} as certifiably recoverable, where 
Alg. \ref{alg:patchcleanser_original} computes  $\mathbb{M}_{min}=\emptyset$ at line 2 and returns the label at line 4.
We know that lines 1--4 in Alg. \ref{alg:patchcleanser_original} and 
Alg. \ref{alg:patchcleanser_revision} are the same.
So, Alg. \ref{alg:patchcleanser_original} and 
Alg. \ref{alg:patchcleanser_revision} return the same label for \textit{x}. 
Moreover, by Lem. \ref{lem:pc_analysis_revised}, we know that \textit{x} is certifiably recoverable by PatchCleanser with Alg. \ref{alg:patchcleanser_revision} for label prediction. 
Since both PatchCleanser with Alg. \ref{alg:patchcleanser_original} and PatchCleanser with Alg. \ref{alg:patchcleanser_revision} use the TMA condition to give a recovery certificate and produce the same label for these certified samples, 
\textit{CrossCert} and \textit{CrossCert-base} share the same certifiably recovery accuracy
 with PatchCleanser with Alg. \ref{alg:patchcleanser_original}.
Furthermore, in both \textit{CrossCert} and \textit{CrossCert-base},   
$c_\textsc{u}(\textit{x})$ is formulated as $\mathbf{1}[g_1(\textit{x})=g_2(\textit{x})\land c_1(\textit{x})=c_2(\textit{x})=\textit{True}]$.
Suppose the two defenders are configured with the same $R_2$.
Then, we know that $c_\textsc{u}(\textit{x})$ in \textit{CrossCert} returns \textit{True} if and only if $c_\textsc{u}(\textit{x})$ in \textit{CrossCert-base} returns \textit{True}.
As such, their unwavering certification outcomes for any benign sample are also the same.

Lem. \ref{lem:pc_analysis} below formulates an important concept. 
It quantifies a conservative set that contains \emph{all} attack configurations (without raising a warning) of a masking-based recovery defender $R_{m}$ for \textit{x}, where an attack configuration indicates where a patch \emph{can} appear in \textit{x} and a possible malicious label produced by $R_{m}$ for the resulting malicious version of \textit{x}.
It formulates the condition $\textit{NAC}_{\textit{masking}}(.)$ to quantify the membership of the set.
Lem. \ref{lem:pc_analysis} states the following: 
If an attack configuration (\textsc{p}, $y_a$) produces a malicious sample to make Alg. \ref{alg:patchcleanser_revision} output a non-benign label $y_a$ in line 4 or 11 (without raising a warning), then the patch region must be cover by two masks (possibly identical), and the mutants produced by the benign version of this malicious sample with these two masks should output $y_a$ by the base model $h$.
The idea of the proof is through a case-by-case analysis. 
With a malicious sample as input, it backtracks the output variable to compute the values generated by masking in order to finally determine the constraints on the benign version of the malicious sample.
After we formulate a similar condition (Lem. \ref{lem:voting_analysis}) for a voting-based recovery defender, we will use both conditions to construct two sets to detect the patch regions with inconsistent label pairs across the two sets to obtain an improved theoretical guarantee for detection certification (Thm. \ref{thm: intersection}), compared to \textit{CrossCert-base}.
\begin{lem}[Necessary Attack Condition for Masking-based Recovery]\label{lem:pc_analysis}
Suppose  $R_m =\langle g, v, c \rangle$ is a masking-based recovery defender, where the prediction function $g$ is Alg. \ref{alg:patchcleanser_revision}. 
Suppose \textit{x} is a benign sample and an attacker wants to manipulate the label of \textit{x}'s malicious version to be output at line 4 (Case I) or line 11 (Case II) in Alg. \ref{alg:patchcleanser_revision} (and thus $v(\textit{x}') = \textit{False}$).
Then, the attack configuration $(\textsc{p},y_a)$ must satisfy the condition 
$\textit{NAC}_{\textit{masking}}(\textit{x}, \textsc{p}, y_a) =$
$[\textsc{p}\odot \textsc{m}_d=\textsc{p}
\land 
y_a=y_d$ where $(\textsc{m}_d,y_d)\in \{(\textsc{m}_d,y_d)\mid \textsc{m}_1, \textsc{m}_2 \in \mathbb{M} \land h((\textsc{J}-\textsc{m}_1-\textsc{m}_2)\odot \textit{x})=y_d\neq g(\textit{x})\}
\land$
 $\textsc{p}\in\mathbb{P} \land \textsc{m}_d=\textsc{m}_1+\textsc{m}_2]$ (Note that we allow $\textsc{m}_1=\textsc{m}_2$.).
\end{lem}
\begin{proof}
\emph{Case 1}: Suppose the benign sample satisfies the TMA condition $[\exists y \in \mathcal{Y}, \forall \textsc{m}_1, \textsc{m}_2 \in \mathbb{M}, h((\textsc{J}-\textsc{m}_1-\textsc{m}_2)\odot \textit{x})=y]$.
In this case, by Lem. \ref{lem:pc_analysis_revised}, no attack configuration is feasible. 
\emph{Case 2}: Suppose the benign sample does not satisfy the TMA condition. 
\emph{Case 2a}: Suppose the mutant set of the malicious sample is self-consistent. 
By the definition of the mask set, there is a mask, denoted by $\textsc{m}_1$, {to cover} the patch (and the patch region) of the malicious sample; therefore, the prediction label of the mutant {for the sample} generated by this mask, denoted by $y_d$, should be the same as the corresponding one {for} \textit{x}.
Let us denote it by $(\textsc{J}-\textsc{m}_1)\odot\textit{x}'$, which is equivalent to $(\textsc{J}-\textsc{m}_1)\odot\textit{x}$. 
All mutants in this set should be predicted to $y_d$ (since this set is self-consistent), which means $y_{\textit{maj}} = y_d$ at line 4. 
Therefore, 
we should have $y_a = y_d$ so that the output label is the label required by the attack configuration. So, the attack configuration (\textsc{p}, $y_a$) should satisfy
the condition $h((\textsc{J}-\textsc{m}_1)\odot\textit{x}) = y_d = y_a \neq g(\textit{x})\land\textsc{m}_1\odot\textsc{p}=\textsc{p}\land \textsc{m}_1\in\mathbb{M}$. 
\emph{Case 2b}: On the other hand, suppose that the mutant set of the malicious sample is not self-consistent.  Then, $\mathbb{M}_{\textit{min}}$ will contain the non-majority mutants at line 2. So, the algorithm goes to line 6 instead of reaching line 4.
The proof is identical to Case 2a except for the following substitutions: 
the term ``mutant set of the malicious sample'' is replaced by the term  ``mutant set of the mutant of the malicious sample'', 
the term ``a mask'' ($\textsc{m}_1$) is replaced by the term  
``{the} union of two masks'' 
($\textsc{m}_1$+
$\textsc{m}_2$),
the variable $y_{\textit{maj}}$ is replaced by  $y_{\textit{maj}}'$, and
the target line for label output is changed from line 4 to line 11.
As such, we can infer that the attack configuration (\textsc{p}, $y_a$) must satisfy
the condition $h((\textsc{J}-\textsc{m}_1-\textsc{m}_2)\odot\textit{x}) = y_d = y_a \neq g(\textit{x})\land(\textsc{m}_1+\textsc{m}_2)\odot\textsc{p}=\textsc{p}\land \textsc{m}_1,\textsc{m}_2\in\mathbb{M}$.
\end{proof}

\subsubsection{Analysis of Voting-Based recovery}
\label{sec:analysis-voting-recovery}
In the last subsection, we have formulated the condition $\textit{NAC}_{\textit{masking}}(.)$ to find an interesting set of attack configurations for a masking-based recovery defender.
 In this subsection, we formulate a similar condition, known as $\textit{NAC}_{\textit{voting}}(.)$, for the same purpose but for a voting-based recovery defender, which is defined in Lem. \ref{lem:voting_analysis}.

Lem. \ref{lem:voting_analysis}  captures the following intuition: 
Suppose \textit{x} is a benign sample predicted to a benign label $y$.
Suppose further that the number of ablations of \textit{x} that do not overlap with a patch region $\textsc{p}$ and are predicted to $y$ is smaller than the number of ablations of \textit{x} that are either predicted to another label $y_a$ ($\neq y$) or overlap with $\textsc{p}$. In this case, $y_a$ is a possible malicious label of patch region $\textsc{p}$. 

 
\begin{lem}[Necessary Attack Condition of Voting-based Recovery]\label{lem:voting_analysis}
Suppose  $R_v=\langle g,v,c \rangle$, where $g(\textit{x})=$
\emph{arg}$max_{y}
|\{\textit{x}_\textsc{b} \in \mathbb{X}_\mathbb{B}(\textit{x}) \mid f(\textit{x}_\textsc{b}) = y\}|$ is 
a voting-based recovery defender.
Given a benign sample \textit{x}, if an attacker wants to manipulate the output label of a specific malicious version of \textit{x}, the attack configuration $(\textsc{p},y_a)$ must satisfy the condition 
$\textit{NAC}_{voting}(\textit{x}, \textsc{p},y_a)$, which is defined as
$[|\{\textit{x}_\textsc{b}\mid\textit{x}_\textsc{b} \in \mathbb{X}_\mathbb{B}(\textit{x})-
\mathbb{X}_\textsc{p}(\textit{x})\land f(\textit{x}_\textsc{b})=y_a\}|+|\mathbb{X}_\textsc{p}(\textit{x})| \geq|\{\textit{x}_\textsc{b}\mid\textit{x}_\textsc{b} \in \mathbb{X}_\mathbb{B}(\textit{x})-
\mathbb{X}_\textsc{p}(\textit{x})\land f(\textit{x}_\textsc{b})=g(\textit{x})\}|]$, where $\mathbb{X}_\textsc{p}(\textit{x})=\{ \textit{x}_\textsc{b} \mid \textit{x}_\textsc{b}=\textit{x}\odot\textsc{b} \land \textsc{b}\odot\textsc{p}\neq \textsc{O} \land\textsc{b}\in\mathbb{B}\}$
is the set of ablations for \textit{x} that overlap with \textsc{p}.
\end{lem}
\begin{proof}
Suppose the patch of a malicious sample is in a patch region \textsc{p} and the intended label of the sample is $y_a$. 
Thus, we obtain $\mathbb{X}_\textsc{p}$. 
In the worst case, the ablations predicted to $y_a$ are either overlap with \textsc{p} (so that the prediction labels are manipulable by the patch) or do not overlap with \textsc{p} but are predicted to $y_a$. 
The number of such ablations is $|\{\textit{x}_\textsc{b} \in \mathbb{X}_\mathbb{B}(\textit{x})-
\mathbb{X}_\textsc{p}(\textit{x}) \mid f(\textit{x}_\textsc{b})=y_a\}|+|\mathbb{X}_\textsc{p}(\textit{x})|$.
To make this set of ablations large enough to lead $g(.)$ to output $y_a$, this number 
should be larger than or equal to (depending on how the argmax function is implemented) the number of ablations that do not overlap with \textsc{p} and are predicted to $g(\textit{x})$, i.e.,
$|\{\textit{x}_\textsc{b} \in \mathbb{X}_\mathbb{B}(\textit{x})-\mathbb{X}_\textsc{p}(\textit{x}) \mid f(\textit{x}_\textsc{b})=g(\textit{x})\}|$,
whose prediction labels cannot be manipulable by the patch.   
\end{proof}


\subsubsection{Theoretical Guarantee of \textit{CrossCert}}\label{sec:guarantee_of_cc}
In this subsection, we link the analysis results and outputs of the base recovery defenders we have established in the last two subsections to cross-check and develop the guarantees for unwavering and detection certifications of \textit{CrossCert}.
Thm. \ref{thm:consistency} and  Thm. \ref{thm: intersection} establish the connections between the conditions for unwavering certification and suppressing warnings, as well as those between the conditions for detection certification and raising warnings, respectively.

\begin{thm}[Unwavering Certification-Warning Verification Consistency]
\label{thm:consistency}
Given a benign sample \textit{x}, 
a masking-based recovery defender $R_1=\langle g_1, v_1, c_1\rangle$ that uses Alg. \ref{alg:patchcleanser_revision} for prediction, and a voting-based recovery defender
$R_2=\langle g_2, v_2,c_2\rangle$,
if $[g_1(\textit{x})=g_2(\textit{x})\land c_1(\textit{x})=c_2(\textit{x})=\textit{True}]$ holds, 
then $[\forall \textit{x}'\in \mathbb{A}(\textit{x}), g_1(\textit{x}')=g_1(\textit{x})\land g_1(\textit{x}')=g_2(\textit{x}')\land v_1(\textit{x}')=v_2(\textit{x}')=\textit{False}]$ holds.
\end{thm}
\begin{proof}
The proof follows the proof of Thm. \ref{thm:theory_1}.
Still, we need to prove that the condition
$[
\forall \textit{x}'\in \mathbb{A}(\textit{x}),
v_1(\textit{x}')=v_2(\textit{x}')=\textit{False}
]$ can be inferred from the condition stated in the theorem.
Since $c_1(\textit{x})$ is \textit{True}, we know that $R_1$ certifies the sample \textit{x}.
By Lem. \ref{lem:pc_analysis_revised}, $R_1$ with Alg. \ref{alg:patchcleanser_revision} for prediction will never reach line 14 with any $\textit{x}'\in \mathbb{A}(\textit{x})$ as input.
Thus, we have $[
\forall \textit{x}'\in \mathbb{A}(\textit{x}),
v_1(\textit{x}')=\textit{False}
]$.
Since we also know that the condition $[v_2(.) = \textit{False}]$ always holds, the result follows.
\end{proof}

An attempt to attack a certified recovery defender may succeed or fail.
Each possible attack configuration (\textsc{p}, $y_a$) for the masking-based recovery defender $R_1$ to output $y_a$ in Cases I and  II will satisfy $\textit{NAC}_{\textit{masking}}(\textit{x}, \textsc{p},y_a)$ by Lem. \ref{lem:pc_analysis}, and those attack configurations for the voting-based recovery defender $R_2$ will satisfy $\textit{NAC}_{\textit{voting}}(\textit{x}, \textsc{p},y_a)$ by Lem. \ref{lem:voting_analysis}.
Also, each failed attack by manipulating a benign sample \textit{x} into a malicious sample for such a defender will produce the benign label of \textit{x}.
Since every label returned in Case III in Alg.~\ref{alg:patchcleanser_revision} for $R_1$ will raise a warning by design, 
we can determine all possible pairs of patch region and target label applicable to a benign sample \textit{x} that \emph{neither} $R_1$ \emph{nor} $R_2$
can issue any warnings.
We define two sets $\mathbb{ATT}_{R_1}(\textit{x})$ and  $\mathbb{ATT}_{R_2}(\textit{x})$ to contain all such possible pairs for $R_1$ and $R_2$
as
$\mathbb{ATT}_{R_1}(\textit{x})=
  \{ 
    (\textsc{p},y_a) \mid \textit{NAC}_{\textit{masking}}(\textit{x}, \textsc{p},y_a)
  \}
  \cup
  \{(\textsc{p},e)\mid\textsc{p}\in\mathbb{P} \land e=g_1(\textit{x})\}
$
and
$\mathbb{ATT}_{R_2}(\textit{x})=
  \{ 
    (\textsc{p},y_a) \mid \textit{NAC}_{\textit{voting}}(\textit{x}, \textsc{p},y_a)
  \}
  \cup
  \{(\textsc{p},e)\mid\textsc{p}\in\mathbb{P} \land e=g_2(\textit{x})\}$%
  $
$,
respectively.

Thm. \ref{thm: intersection} formalizes the intuition of not making any common mistake for detection: Suppose these two sets share no common pair of patch region and target label that can manipulate \textit{x} to cause the corresponding two base recovery defenders to produce the same malicious label. 
In this case, they cannot predict the same label for the same malicious version of \textit{x} unless the label is benign.


\begin{thm}[Detection Certification-Warning Verification consistency]\label{thm: intersection}
Given a sample \textit{x}, two recovery defenders $R_1=\langle g_1, v_1, c_1\rangle$ and $R_2=\langle g_2, v_2, c_2\rangle$ as defined in Thm. \ref{thm:consistency},
along with their $\mathbb{ATT}_{R_1}(\textit{x})$ and $\mathbb{ATT}_{R_2}(\textit{x})$. 
If $[\mathbb{ATT}_{R_1}(\textit{x})\cap\mathbb{ATT}_{R_2}(\textit{x})=\emptyset]\lor[\forall (\textsc{p},y_a)\in\mathbb{ATT}_{R_1}(\textit{x})\cap\mathbb{ATT}_{R_2}(\textit{x}), y_a=g_1(\textit{x})]$, 
then $[\forall \textit{x}' \in \mathbb{A}(\textit{x})$, $g_1(\textit{x})\neq g_1(\textit{x}') \implies g_1(\textit{x}')\neq g_2(\textit{x}')\lor v_1(\textit{x}')=\textit{True}\lor v_2(\textit{x}')=\textit{True}]$.
\end{thm}
\begin{proof}
We first recall that any malicious sample 
corresponds to one pair of a patch region and a label for ${R_1}$ and another pair for ${R_2}$ such that these two pairs share the same patch region. 
%
%
\emph{Case 1}: Suppose that $[\mathbb{ATT}_{R_1}(\textit{x})\cap\mathbb{ATT}_{R_2}(\textit{x})=\emptyset]$ holds. 
It implies that no common pair $(\textsc{p}, e)$ $\in \mathbb{ATT}_{R_1}(\textit{x})\cap\mathbb{ATT}_{R_2}(\textit{x})$ exists to make both defenders produce the same prediction label. 
There are two subcases to consider.
\emph{Case 1a}:
Suppose $[(\textsc{p}, g_1(\textit{x}')) \in \mathbb{ATT}_{R_1}(\textit{x}) \land (\textsc{p}, g_2(\textit{x}')) \in \mathbb{ATT}_{R_2}(\textit{x})]$ holds.
As there is no common pair between $\mathbb{ATT}_{R_1}(\textit{x})$ and  $\mathbb{ATT}_{R_2}(\textit{x})$,  
%
we have
 $g_1(\textit{x}')\neq g_2(\textit{x}')$.
\emph{Case 1b}:
Suppose $[(\textsc{p}, g_1(\textit{x}')) \notin \mathbb{ATT}_{R_1}(\textit{x}) 
\lor (\textsc{p}, g_2(\textit{x}')) \notin \mathbb{ATT}_{R_2}(\textit{x})]$ holds. 
\emph{Case 1b(i)}: Suppose $(\textsc{p}, g_1(\textit{x}')) \notin \mathbb{ATT}_{R_1}$ holds. By the definition of $\mathbb{ATT}_{R_1}(\textit{x})$, neither $NAC_{masking}(\textit{x}, \textsc{p}, g_1(\textit{x}'))$ nor $g_1(\textit{x}')=g_1(\textit{x})$ hold.
By Lem.~\ref{lem:pc_analysis}, Alg. \ref{alg:patchcleanser_revision} reaches Case III, and $v_1(\textit{x}')$ returns \textit{True}. 
\emph{Case 1b(ii)}: Suppose $(\textsc{p}, g_2(\textit{x}')) \notin \mathbb{ATT}_{R_2}$ holds. By the definition of $\mathbb{ATT}_{R_2}(\textit{x})$,
neither $NAC_{voting}(\textit{x}, \textsc{p}, g_2(\textit{x}'))$ nor $g_2(\textit{x}')=g_2(\textit{x})$ hold. 
{Since} $g_2(\textit{x}')\neq g_2(\textit{x})$, by Lem. \ref{lem:voting_analysis}, the condition 
$NAC_{voting}(\textit{x}, \textsc{p}, g_2(\textit{x}'))$ must hold.
There is a contradiction, indicating that this subcase is infeasible. 
So, for 
Case 1b, the condition
$[v_1(\textit{x}')=\textit{True}\lor v_2(\textit{x}')=\textit{True}]$ must hold.
\emph{Case 2}: Suppose that the condition 
$[\forall (\textsc{p},y_a)\in\mathbb{ATT}_{R_1}(\textit{x})\cap\mathbb{ATT}_{R_2}(\textit{x}), y_a=g_1(\textit{x})]$ holds.
\emph{Case 2a}: 
Suppose $[(\textsc{p}, g_1(\textit{x}')) \in \mathbb{ATT}_{R_1}(\textit{x}) \land (\textsc{p}, g_2(\textit{x}')) \in \mathbb{ATT}_{R_2}(\textit{x})]$ holds.
There are two subcases in Case 2a for consideration. 
\emph{Case 2a(i)}: Suppose  $(\textsc{p}, g_1(\textit{x}')) = (\textsc{p}, g_2(\textit{x}'))$.
In this case, this common pair must reside in $\mathbb{ATT}_{R_1}(\textit{x})\cap\mathbb{ATT}_{R_2}(\textit{x})$.
So, we must have 
$g_1(\textit{x}') = g_2(\textit{x}') = g_1(\textit{x})$, indicating that {the condition $[g_1(\textit{x})\neq g_1(\textit{x}')]$ in the antecedent stated in the theorem} does not hold. 
\emph{Case 2a(ii)}: Suppose $(\textsc{p}, g_1(\textit{x}')) \neq (\textsc{p}, g_2(\textit{x}'))$. 
We must have  $g_1(\textit{x}') \neq g_2(\textit{x}')$.
\emph{Case 2b}: Suppose $[(\textsc{p}, g_1(\textit{x}')) \notin \mathbb{ATT}_{R_1}(\textit{x}) 
\lor (\textsc{p}, g_2(\textit{x}')) \notin \mathbb{ATT}_{R_2}(\textit{x})]$ holds. 
The proof for Case 2b is identical to the proof for Case 1b.
\end{proof}

\subsubsection{Soundness and Completeness}
The two types of certification achieved by \textit{CrossCert} are sound (without false positives) and incomplete (with false positives) when certifying samples.

In Thm. \ref{thm:consistency}, the condition in the consequence component is
$[\forall \textit{x}'\in \mathbb{A}(\textit{x}), g_1(\textit{x}')=g_1(\textit{x})\land g_1(\textit{x}')=g_2(\textit{x}')\land v_1(\textit{x}')=v_2(\textit{x}')=\textit{False}]$.
It can be \emph{rewritten} into 
$[\forall \textit{x}'\in \mathbb{A}(\textit{x}), 
g_1(\textit{x}')=g_1(\textit{x})\land 
\neg [\neg 
[g_1(\textit{x}')=g_2(\textit{x}')\land v_1(\textit{x}')=v_2(\textit{x}')=\textit{False}]]
]$
$\Leftrightarrow$
$[\forall \textit{x}'\in \mathbb{A}(\textit{x}), 
g_1(\textit{x}')=g_1(\textit{x})\land 
\neg  
[g_1(\textit{x}')\neq g_2(\textit{x}')\lor v_1(\textit{x}')=\textit{True} \lor v_2(\textit{x}')=\textit{True}]]$
$\Leftrightarrow$
$[\forall \textit{x}'\in \mathbb{A}(\textit{x}), 
g_1(\textit{x}') = g_1(\textit{x}) \land 
\neg v(\textit{x}')]$(by the definition of $v(\textit{x}')$ of \textit{CrossCert}), which is the condition in Def.~\ref{def:cert_no_warning}
to quantify a certifiably unwavering sample.
Similarly, in Thm. \ref{thm: intersection}, the condition in the consequence component can be \emph{rewritten} into 
$[\forall \textit{x}'\in \mathbb{A}(\textit{x}), 
g_1(\textit{x}') \neq g_1(\textit{x}) \Rightarrow 
v(x')]
$, which is the condition in Def.~\ref{def:cert_detect}
to quantify a certifiably detectable sample.

For some specific samples, the actual lower bound needs not be as loose as the worst case.  
For instance, if exhausting all malicious versions of a specific sample that \textit{CrossCert} fails to certify is possible, the certification on the sample may still be possible.
Besides, the recovery certification of \textit{CrossCert} 
is sound and incomplete. 
Lem.~\ref{lem:pc_analysis_revised} has proven its soundness.
\textit{CrossCert} 
needs the TMA condition to hold to prove Lem. ~\ref{lem:pc_analysis_revised}, whereas Def. \ref{def:cert_recovery} does not need this condition.

\section{Evaluation}\label{sec:eva}

\subsection{Research Questions}
We aim to answer the following research questions through experiments.
\begin{itemize}
\item[RQ1] How does \textit{CrossCert} perform compared to state-of-the-art certified detection techniques?
\item[RQ2] To what extent does \textit{CrossCert} improve the effectiveness compared to \textit{CrossCert-base}?
\end{itemize}

\subsection{Experimental Setup}
\subsubsection{Environment}

All the code of our defenders and the experiments are implemented in Python 3.8 using the PyTorch 2.0.1 framework. 
We run all the experiments on a machine equipped with the Ubuntu 20.04 operating system with four 2080Ti GPU cards. 

\subsubsection{Datasets}
We adopt three popular image classification datasets including 1000-class ImageNet \cite{deng2009imagenet}, 100-class CIFAR100 \cite{krizhevsky2009learning} and 10-class CIFAR10 \cite{krizhevsky2009learning} as our evaluation datasets, which are widely adopted in previous patch robustness certification research \cite{xiang2021patchguard++,xiang2022patchcleanser, patchcensor, li2022vip, han2021scalecert}. 
ImageNet contains 1.3M training images and 50k validation images for 1000 classes, 
CIFAR100 contains 50k training images and 10k test images for 100 classes, and
CIFAR10 contains 50k training images and 10k test images for 10 classes. We download ImageNet from \url{image-net.org}, use its entire training set for fine-tuning, and regard its validation set as the test set for evaluation. We download CIFAR10 and CIFAR100 from torchvision \cite{torchvision2016} and use their whole training sets for fine-tuning and their test sets for evaluation. All images are resized to $224\times224$ in our experiments. 

\subsubsection{Baselines}
\label{sec:baselines}
We compare three top-performing certified detection defenders implemented in our infrastructure with \textit{CrossCert} (\textbf{CC}): Minority Reports Adapted (\textbf{MR+}) \cite{patchcensor}, PatchCensor (\textbf{PC}) \cite{patchcensor}, and \textbf{ViP} \cite{li2022vip}.
Specifically, we adopt two of the most popular 
architectures, Vision Transformer \cite{dosovitskiy2021an} (ViT-b16-224 with 86.6M parameters) and ResNet \cite{he2016deep} (ResNet-50 with 25.5M parameters), as the architectures of the base models of these defenders.
We obtain their architectures and pre-trained weights on ImageNet-21k \cite{deng2009imagenet} from \cite{vit-b16-224} and \cite{resnet-50}, respectively.
Following the main experiments in their papers, we use ResNet-50 for MR+ and ViT-b16-224 for PC and ViP as their base models.

For fine-tuning ResNet and ViT for MR+ and PC, we use SGD with a momentum of 0.9, set the batch size to 64 and the training epoch to 10, and reduce the learning rate by a factor of 10 after every 5 epochs with an initial learning rate of 0.001 for ViT and 0.01 for ResNet, following the fine-tuning process of vanilla models in \cite{xiang2022patchcleanser}.
ViP \cite{li2022vip} adopts the MAE model as its base model, which is the architecture of ViT-b16-224 pretrained and fine-tuned by the MAE's novel training scheme \cite{he2022masked}. 
%
We adopt the pretrained weights of ViT-b16-224 and fine-tuning scripts from the MAE official implementation \cite{MAE_code} and fine-tune the model for each of our three datasets following the default setting in the MAE official script.
We use the original implementations of ViP \cite{ViP_code} and PC/MR+ \cite{PatchCensor_code} to compute their certifiably detectable accuracy {and clean accuracy (see Section} \ref{sec:metrics}).

We further compare \textit{CrossCert} (CC) and \textit{CrossCert-base} (\textbf{CC-base}) with more state-of-the-art certified detection defenders, and  annotate them with the symbol $\star$: Minority Reports (\textbf{MR$_\star$}) \cite{mccoyd2020minority}, PatchGaurd++ (\textbf{PG++$_\star$}) \cite{xiang2021patchguard++}, ScaleCert (\textbf{SC$_\star$}) \cite{han2021scalecert}, Adapted Minority Reports (\textbf{MR+$_\star$}) \cite{patchcensor}, PatchCensor (\textbf{PC$_\star$}) \cite{patchcensor}, and ViP (\textbf{ViP$_\star$}) \cite{li2022vip} based on the reported results in the literature. 

\subsubsection{Implementation of \textit{CrossCert}}\label{sec:implement}
We adopt the same pretrained models (ViT-b16-224 and ResNet-50) from timm \cite{rw2019timm} stated in Section \ref{sec:baselines} as the {pretrained} base models of the base recovery defenders $R_1$ and $R_2$, to configure CC
into two instances \textbf{CC-ViT} and \textbf{CC-RN}, respectively, and  configure CC-base
into another two instances \textbf{CC-base-ViT} and \textbf{CC-base-RN}, respectively.
%
%
%
We adopt different state-of-the-art recovery defenders as our $R_1$ and $R_2$ because \textit{CrossCert-base} and \textit{CrossCert} use their diversity to gain certifiably detectable accuracy.
We adopt PatchCleanser \cite{xiang2022patchcleanser} with Alg. \ref{alg:patchcleanser_revision} for  prediction as our $R_1$
and
implement
our voting-based recovery defender following the common procedure described \cite{levine2020randomized,salman2022certified,li2022vip} as our $R_2$.
We fine-tune the base model of $R_1$ following the setting described in \cite{xiang2022patchcleanser}.
We also fine-tune the base model of $R_2$ with hyperparameters stated in \cite{salman2022certified} for ImageNet and CIFAR10, and adopt those for CIFAR10 to CIFAR100.


\subsubsection{Metrics}
\label{sec:metrics}
Suppose $\textit{x}$ is a sample with the ground truth $y$ in a test dataset $\mathbb{D}$.
A sample $\textit{x}$ is said to be correct iff the prediction of a defender under test for $x$ is $y$, i.e., $g(\textit{x})$ = $y$.
\textit{Clean accuracy} is the fraction of correct samples:
$acc_{\textit{clean}}=\frac{\mid\{\textit{x}\in\mathbb{D}\mid g(\textit{x})=y\}\mid}{\mid\mathbb{D}\mid}$.

We use the three predicates
$c_u(\textit{x})$, $c_d(\textit{x})$, and $c_r(\textit{x})$ returned by a defender to represent whether the defender certifies $\textit{x}$ to be
certifiably unwavering, certifiably detectable, and
certifiably recoverable, respectively.
Their certified accuracy
corresponds to
the fractions of correct samples that are 
certifiably unwavering, certifiably detectable, and 
certifiably recoverable, defined as
$acc_{\textit{cert}_u}=\frac{\mid\{\textit{x}\in\mathbb{D}\mid g(\textit{x})=y\land c_{u}(\textit{x})\}\mid}{\mid\mathbb{D}\mid}$,
$acc_{\textit{cert}_d}=\frac{\mid\{\textit{x}\in\mathbb{D}\mid g(\textit{x})=y\land c_d(\textit{x})\}\mid}{\mid\mathbb{D}\mid}$, and
$acc_{\textit{cert}_r}=\frac{\mid\{\textit{x}\in\mathbb{D}\mid g(\textit{x})=y\land c_r(\textit{x})\}\mid}{\mid\mathbb{D}\mid}$, respectively.

\subsubsection{Experimental Procedure}

We adopt the frequently-used patch sizes of 1\%, 2\%, and 3\% for ImageNet and 2\% and 2.4\% for CIFAR10 and CIFAR100 (where a patch region is a square) \cite{li2022vip,patchcensor,xiang2021patchguard++,han2021scalecert}.
%
For each combination of patch size and test sample in each test dataset to be certified by each of 
CC-base-ViT, CC-base-RN, CC-ViT, and CC-RN,
we first run the mutant generation scripts of $R_1$  and $R_2$ to generate their mutants
and obtain the prediction labels of all these mutants.
We then run the recovery certification and label prediction scripts 
to determine whether the sample
is certifiably recoverable for each of $R_1$ and $R_2$, and to obtain the prediction label of the sample produced by each of $R_1$ and $R_2$.
%
%
We next compute whether the sample is certifiably unwavering, whether it is certifiably detectable, and whether it is certifiably recoverable for each of CC-ViT, CC-RN, CC-base-ViT, and CC-base-RN.
We compute the metrics stated in Section \ref{sec:metrics}.

For each of MR+, PC, and ViP, like the above procedure, for each combination of patch size and test dataset, we apply the corresponding 
scripts in their official repositories to produce mutants, prediction labels, and certification outcomes for each test sample.

\begin{table}[tb]
\caption{Results of Certified Detection Techniques on CIFAR10}\label{tab:cifar10}
\resizebox{0.7\textwidth}{!}{
\begin{tabular}{|c|ccc|clc|}
\hline
Patch Size & \multicolumn{3}{c|}{0.4\% pixels}                                                              & \multicolumn{3}{c|}{2.4\% pixels}                                                              \\ \hline
Metrics    & \multicolumn{1}{c|}{$acc_{\textit{clean}}$} & \multicolumn{1}{c|}{$acc_{{cert}_d}$}      & $acc_{{cert}_u}$      & \multicolumn{1}{c|}{$acc_{\textit{clean}}$}   & \multicolumn{1}{c|}{$acc_{{cert}_d}$}      & $acc_{{cert}_u}$ \\ \hline
MR$_\star$         & \multicolumn{1}{c|}{87.60\%}            & \multicolumn{1}{c|}{82.50\%}        &       $\oplus$         & \multicolumn{1}{c|}{84.20\%}            & \multicolumn{1}{c|}{78.10\%}        &        $\oplus$        \\ \hline
PG++$_\star$       & \multicolumn{1}{c|}{82.00\%}            & \multicolumn{1}{c|}{78.80\%}        &       $\oplus$         & \multicolumn{1}{c|}{78.20\%}            & \multicolumn{1}{c|}{74.10\%}        &        $\oplus$        \\ \hline
SC$_\star$         & \multicolumn{1}{c|}{83.10\%}            & \multicolumn{1}{c|}{81.00\%}        &       $\oplus$         & \multicolumn{1}{c|}{78.90\%}            & \multicolumn{1}{c|}{75.30\%}        &        $\oplus$        \\ \hline
MR+$_\star$        & \multicolumn{1}{c|}{97.61\%}            & \multicolumn{1}{c|}{90.59\%}        &       $\oplus$         & \multicolumn{1}{c|}{97.70\%}            & \multicolumn{1}{c|}{83.30\%}        &        $\oplus$        \\ \hline
PC$_\star$         & \multicolumn{1}{c|}{98.77\%}            & \multicolumn{1}{c|}{95.70\%}        &       $\oplus$         & \multicolumn{1}{c|}{\textbf{98.81\%}}           & \multicolumn{1}{c|}{88.29\%}        &        $\oplus$        \\ \hline \hline
MR+        & \multicolumn{1}{c|}{97.64\%}            & \multicolumn{1}{c|}{87.92\%}        &       $\oplus$         & \multicolumn{1}{c|}{97.64\%}            & \multicolumn{1}{c|}{72.26\%}        &        $\oplus$        \\ \hline
PC         & \multicolumn{1}{c|}{98.72\%}            & \multicolumn{1}{c|}{95.73\%}        &       $\oplus$         & \multicolumn{1}{c|}{98.72\%}            & \multicolumn{1}{c|}{90.22\%}        &        $\oplus$        \\ \hline
ViP        & \multicolumn{1}{c|}{98.74\%}            & \multicolumn{1}{c|}{\textbf{97.59\%}}        &       $\oplus$         & \multicolumn{1}{c|}{98.74\%}            & \multicolumn{1}{c|}{\textbf{94.71\%}}        &        $\oplus$        \\ \hline
CC         & \multicolumn{1}{c|}{\textbf{98.97\%}}            & \multicolumn{1}{c|}{96.14\%}        & \textbf{78.75\%}        & \multicolumn{1}{c|}{98.70\%}            & \multicolumn{1}{c|}{91.57\%}        & \textbf{67.73\%}        \\ \hline
\end{tabular}
}
\end{table}

In addition, we also follow a best practice in the area of certified detection to extract the ImageNet and CIFAR10 results of MR+$_\star$ and PC$_\star$ from \cite{patchcensor}, ViP$_\star$ from \cite{li2022vip}, the ImageNet results of SC$_\star$, PG++$_\star$ and MR$_\star$ from \cite{li2022vip} and their CIFAR10 results from \cite{han2021scalecert} for comparison. We note that the literature does not report results on CIFAR100.

\begin{table}[]
\caption{Results of Certified Detection Techniques on ImageNet}\label{tab:imagenet}
\resizebox{\textwidth}{!}{
\begin{tabular}{|c|ccccccccc|}
\hline
Patch Size & \multicolumn{3}{c|}{1\% pixels}                                                                                              & \multicolumn{3}{c|}{2\% pixels}                                                                                              & \multicolumn{3}{c|}{3\% pixels}                                                                         \\ \hline
Metrics & \multicolumn{1}{c|}{$acc_{\textit{clean}}$} & \multicolumn{1}{c|}{$acc_{{cert}_d}$}      & \multicolumn{1}{c|}{$acc_{{cert}_u}$}     & \multicolumn{1}{c|}{$acc_{\textit{clean}}$} & \multicolumn{1}{c|}{$acc_{{cert}_d}$}      & \multicolumn{1}{c|}{$acc_{{cert}_u}$}      & \multicolumn{1}{c|}{$acc_{\textit{clean}}$} & \multicolumn{1}{c|}{$acc_{{cert}_d}$}      & \multicolumn{1}{c|}{$acc_{{cert}_u}$}      \\ \hline
MR$_\star$         & \multicolumn{9}{c|}{computationally expensive\tablefootnote{As stated in \cite{han2021scalecert}, MR is very computationally expensive for high-resolution images, taking almost a month for ImageNet.}}                                   \\ \hline
PG++$_\star$       & \multicolumn{1}{c|}{61.80\%}        & \multicolumn{1}{c|}{36.30\%}              & \multicolumn{1}{c|}{$\oplus$} & \multicolumn{1}{c|}{61.60\%}        & \multicolumn{1}{c|}{33.90\%}              & \multicolumn{1}{c|}{$\oplus$} & \multicolumn{1}{c|}{61.50\%}        & \multicolumn{1}{c|}{31.10\%}              & $\oplus$ \\ \hline
SC$_\star$         & \multicolumn{1}{c|}{62.80\%}        & \multicolumn{1}{c|}{60.40\%}              & \multicolumn{1}{c|}{$\oplus$} & \multicolumn{1}{c|}{58.50\%}        & \multicolumn{1}{c|}{55.40\%}              & \multicolumn{1}{c|}{$\oplus$} & \multicolumn{1}{c|}{56.40\%}        & \multicolumn{1}{c|}{52.80\%}              & $\oplus$ \\ \hline
MR+$_\star$\textsuperscript{\ref{same_data}}        & \multicolumn{1}{c|}{75.51\%}        & \multicolumn{1}{c|}{56.31\%}              & \multicolumn{1}{c|}{$\oplus$} & \multicolumn{1}{c|}{75.51\%}        & \multicolumn{1}{c|}{56.31\%}              & \multicolumn{1}{c|}{$\oplus$} & \multicolumn{1}{c|}{75.49\%}        & \multicolumn{1}{c|}{51.23\%}              & $\oplus$ \\ \hline
PC$_\star$\textsuperscript{\ref{same_data}}         & \multicolumn{1}{c|}{82.73\%}        & \multicolumn{1}{c|}{69.41\%}              & \multicolumn{1}{c|}{$\oplus$} & \multicolumn{1}{c|}{82.73\%}        & \multicolumn{1}{c|}{69.41\%}              & \multicolumn{1}{c|}{$\oplus$} & \multicolumn{1}{c|}{82.67\%}        & \multicolumn{1}{c|}{64.97\%}              & $\oplus$ \\ \hline
ViP$_\star$\textsuperscript{\ref{same_data}}        & \multicolumn{1}{c|}{\textbf{83.66\%}}        & \multicolumn{1}{c|}{\textbf{74.56\%}}              & \multicolumn{1}{c|}{$\oplus$} & \multicolumn{1}{c|}{\textbf{83.66\%}}        & \multicolumn{1}{c|}{\textbf{74.56\%}}              & \multicolumn{1}{c|}{$\oplus$} & \multicolumn{1}{c|}{\textbf{83.66\%}}        & \multicolumn{1}{c|}{\textbf{70.90\%}}              & $\oplus$ \\ \hline
\hline
MR+\tablefootnote{\label{same_data}The best results for using their masking strategy are shown. We note that their results for the patch region size of 1\% are the same as those of 2\% according to the formulations in their defenders.}         & \multicolumn{1}{c|}{80.71\%}        & \multicolumn{1}{c|}{56.34\%}              & \multicolumn{1}{c|}{$\oplus$} & \multicolumn{1}{c|}{80.71\%}        & \multicolumn{1}{c|}{56.34\%}              & \multicolumn{1}{c|}{$\oplus$} & \multicolumn{1}{c|}{80.71\%}        & \multicolumn{1}{c|}{50.29\%}              & $\oplus$ \\ \hline
PC\textsuperscript{\ref{same_data}}         & \multicolumn{1}{c|}{81.83\%}        & \multicolumn{1}{c|}{66.21\%}              & \multicolumn{1}{c|}{$\oplus$} & \multicolumn{1}{c|}{81.83\%}        & \multicolumn{1}{c|}{66.21\%}              & \multicolumn{1}{c|}{$\oplus$} & \multicolumn{1}{c|}{81.83\%}        & \multicolumn{1}{c|}{61.01\%}              & $\oplus$ \\ \hline
ViP\textsuperscript{\ref{same_data}}        & \multicolumn{1}{c|}{82.80\%}        & \multicolumn{1}{c|}{73.98\%}              & \multicolumn{1}{c|}{$\oplus$} & \multicolumn{1}{c|}{82.80\%}        & \multicolumn{1}{c|}{73.98\%}              & \multicolumn{1}{c|}{$\oplus$} & \multicolumn{1}{c|}{82.80\% }        & \multicolumn{1}{c|}{70.42\%}              & $\oplus$ \\ \hline
CC         & \multicolumn{1}{c|}{81.73\%}        & \multicolumn{1}{c|}{67.84\%}              & \multicolumn{1}{c|}{\textbf{44.68\%}}               & \multicolumn{1}{c|}{81.79\%}        & \multicolumn{1}{c|}{63.15\%}              & \multicolumn{1}{c|}{\textbf{39.52\%}}               & \multicolumn{1}{c|}{81.55\%}        & \multicolumn{1}{c|}{59.51\%}              & \textbf{35.48\%}               \\ \hline
\end{tabular}
}
\end{table}

\begin{table}[]
\caption{Results of Certified Detection Techniques on CIFAR100}\label{tab:cifar100}
\resizebox{0.7\textwidth}{!}{
\begin{tabular}{|c|ccc|ccc|}
\hline
Patch Size & \multicolumn{3}{c|}{0.4\% pixels}                                                              & \multicolumn{3}{c|}{2.4\% pixels}                                                              \\ \hline
Metrics    & \multicolumn{1}{c|}{$acc_{\textit{clean}}$} & \multicolumn{1}{c|}{$acc_{{cert}_d}$}      & $acc_{{cert}_u}$      & \multicolumn{1}{c|}{$acc_{\textit{clean}}$} & \multicolumn{1}{c|}{$acc_{{cert}_d}$}      & $acc_{{cert}_u}$      \\ \hline
MR+        & \multicolumn{1}{c|}{88.09\%}            & \multicolumn{1}{c|}{64.51\%}        & $\oplus$         & \multicolumn{1}{c|}{88.09\%}            & \multicolumn{1}{c|}{44.11\%}        & $\oplus$         \\ \hline
PC         & \multicolumn{1}{c|}{{92.64\%}}           & \multicolumn{1}{c|}{83.32\%}        & $\oplus$         & \multicolumn{1}{c|}{\textbf{92.64\%}}           & \multicolumn{1}{c|}{72.71\%}        & $\oplus$         \\ \hline
ViP        & \multicolumn{1}{c|}{89.74\%}            & \multicolumn{1}{c|}{\textbf{85.37\%}}        & $\oplus$         & \multicolumn{1}{c|}{89.74\%}            & \multicolumn{1}{c|}{\textbf{78.28\%}}        & $\oplus$         \\ \hline
CC         & \multicolumn{1}{c|}{\textbf{92.79\%}}            & \multicolumn{1}{c|}{84.01\%}        & \textbf{52.50\%}        & \multicolumn{1}{c|}{{92.48\%}}          & \multicolumn{1}{c|}{73.18\%}        & \textbf{39.07\%}        \\ \hline
\end{tabular}
}
\end{table}

\begin{table}[]
\caption{Results of \textit{CrossCert} and \textit{CrossCert-base} with Different Base Models on the Three Datasets 
}
\label{tab:cc-base}
\resizebox{\textwidth}{!}{
\begin{tabular}{|lll|lll|ll|ll|}
\hline
\multicolumn{3}{|c|}{Dataset}  & \multicolumn{3}{c|}{ImageNet} & \multicolumn{2}{c|}{CIFAR100}          & \multicolumn{2}{c|}{CIFAR10}            \\ \hline
\multicolumn{3}{|c|}{Patch Size (in pixels)}  &  \multicolumn{1}{c|}{1\%}     & \multicolumn{1}{c|}{2\%}     & \multicolumn{1}{c|}{3\%}      & \multicolumn{1}{c|}{0.4\%}   & \multicolumn{1}{c|}{2.4\%}    & \multicolumn{1}{c|}{0.4\%}   & \multicolumn{1}{c|}{2.4\%}    \\ \hline
\multicolumn{1}{|c|}{\multirow{10}{*}{\rotatebox{90}{ViT}}}    & \multicolumn{1}{c|}{\multirow{4}{*}{CC}}      & $acc_{\textit{clean}}$       &
\multicolumn{1}{l|}{81.73\%} & \multicolumn{1}{l|}{81.79\%} & 81.55\% & \multicolumn{1}{l|}{92.79\%} & 92.48\% & \multicolumn{1}{l|}{98.97\%} & 98.70\% \\ \cline{3-10} 
\multicolumn{1}{|c|}{}                         & \multicolumn{1}{l|}{}                         & $acc_{{cert}_d}$   & \multicolumn{1}{l|}{\hspace{1ex}67.84\%} & \multicolumn{1}{l|}{\hspace{1ex}63.15\%} & \hspace{1ex}59.51\% & \multicolumn{1}{l|}{\hspace{1ex}84.01\%} & \hspace{1ex}73.18\% & \multicolumn{1}{l|}{\hspace{1ex}96.14\%} & \hspace{1ex}91.57\% \\ \cline{3-10} 
\multicolumn{1}{|c|}{}                         & \multicolumn{1}{l|}{}                         & $acc_{{cert}_u}$  & \multicolumn{1}{l|}{\hspace{2ex}44.68\%} & \multicolumn{1}{l|}{\hspace{2ex}39.52\%} & \hspace{2ex}35.48\% & \multicolumn{1}{l|}{\hspace{2ex}52.50\%} & \hspace{2ex}39.07\% & \multicolumn{1}{l|}{\hspace{2ex}78.75\%} & \hspace{2ex}67.73\% \\ \cline{3-10} 
\multicolumn{1}{|c|}{}                         & \multicolumn{1}{l|}{}                         & $acc_{{cert}_r}$    & \multicolumn{1}{l|}{\hspace{3ex}61.86\%} & \multicolumn{1}{l|}{\hspace{3ex}57.83\%} & \hspace{3ex}54.63\% & \multicolumn{1}{l|}{\hspace{3ex}78.40\%} & \hspace{3ex}68.57\% & \multicolumn{1}{l|}{\hspace{3ex}94.03\%} & \hspace{3ex}88.71\% \\ \cline{2-10} 
\multicolumn{1}{|c|}{}                         & \multicolumn{1}{c|}{PL-v2}                         &  $acc_{{cert}_r}$ & \multicolumn{1}{l|}{\hspace{3ex}61.86\%} & \multicolumn{1}{l|}{\hspace{3ex}57.83\%} & \hspace{3ex}54.63\% & \multicolumn{1}{l|}{\hspace{3ex}78.40\%} & \hspace{3ex}68.57\% & \multicolumn{1}{l|}{\hspace{3ex}94.03\%} & \hspace{3ex}88.71\% \\ \cline{2-10} 
\multicolumn{1}{|c|}{}                         & \multicolumn{1}{c|}{\multirow{5}{*}{CC-base}} & $acc_{\textit{clean}}$       & \multicolumn{1}{l|}{81.73\%} & \multicolumn{1}{l|}{81.79\%} & 81.55\% & \multicolumn{1}{l|}{92.79\%} & 92.48\% & \multicolumn{1}{l|}{98.97\%} & 98.70\% \\ \cline{3-10} 
\multicolumn{1}{|c|}{}                         & \multicolumn{1}{l|}{}                         & $acc_{{cert}_d}$   & \multicolumn{1}{l|}{\hspace{1ex}65.19\%} & \multicolumn{1}{l|}{\hspace{1ex}61.07\%} & \hspace{1ex}57.84\% & \multicolumn{1}{l|}{\hspace{1ex}80.85\%} & \hspace{1ex}71.32\% & \multicolumn{1}{l|}{\hspace{1ex}95.57\%} & \hspace{1ex}90.75\% \\ \cline{3-10} 
\multicolumn{1}{|c|}{}                         & \multicolumn{1}{l|}{}                         & $acc_{{cert}_u}$  & \multicolumn{1}{l|}{\hspace{2ex}44.68\%} & \multicolumn{1}{l|}{\hspace{2ex}39.52\%} & \hspace{2ex}35.48\% & \multicolumn{1}{l|}{\hspace{2ex}52.50\%} & \hspace{2ex}39.07\% & \multicolumn{1}{l|}{\hspace{2ex}78.75\%} & \hspace{2ex}67.73\% \\ \cline{3-10} 
\multicolumn{1}{|c|}{}                         & \multicolumn{1}{l|}{}                         & $acc_{{cert}_r}$    & \multicolumn{1}{l|}{\hspace{3ex}61.86\%} & \multicolumn{1}{l|}{\hspace{3ex}57.83\%} & \hspace{3ex}54.63\% & \multicolumn{1}{l|}{\hspace{3ex}78.40\%} & \hspace{3ex}68.57\% & \multicolumn{1}{l|}{\hspace{3ex}94.03\%} & \hspace{3ex}88.71\% \\ \cline{2-10} 
\multicolumn{1}{|c|}{}                         & \multicolumn{1}{c|}{PL-v1}                         &  $acc_{{cert}_r}$ & \multicolumn{1}{l|}{\hspace{3ex}61.86\%} & \multicolumn{1}{l|}{\hspace{3ex}57.83\%} & \hspace{3ex}54.63\% & \multicolumn{1}{l|}{\hspace{3ex}78.40\%} & \hspace{3ex}68.57\% & \multicolumn{1}{l|}{\hspace{3ex}94.03\%} & \hspace{3ex}88.71\% \\ \hline
\multicolumn{1}{|c|}{\multirow{10}{*}{\rotatebox{90}{ResNet}}} & \multicolumn{1}{c|}{\multirow{5}{*}{CC}}      & $acc_{\textit{clean}}$       & \multicolumn{1}{l|}{80.69\%} & \multicolumn{1}{l|}{80.74\%} & 80.57\% & \multicolumn{1}{l|}{88.64\%} & 88.25\% & \multicolumn{1}{l|}{97.75\%} & 97.81\% \\ \cline{3-10} 
\multicolumn{1}{|c|}{}                         & \multicolumn{1}{l|}{}                         & $acc_{{cert}_d}$   & \multicolumn{1}{l|}{\hspace{1ex}63.43\%} & \multicolumn{1}{l|}{\hspace{1ex}57.76\%} & \hspace{1ex}53.94\% & \multicolumn{1}{l|}{\hspace{1ex}78.74\%} & \hspace{1ex}60.62\% & \multicolumn{1}{l|}{\hspace{1ex}91.89\%} & \hspace{1ex}84.23\% \\ \cline{3-10} 
\multicolumn{1}{|c|}{}                         & \multicolumn{1}{l|}{}                         & $acc_{{cert}_u}$  & \multicolumn{1}{l|}{\hspace{2ex}35.84\%} & \multicolumn{1}{l|}{\hspace{2ex}30.44\%} & \hspace{2ex}26.42\% & \multicolumn{1}{l|}{\hspace{2ex}14.29\%} &\hspace{2ex}8.65\%  & \multicolumn{1}{l|}{\hspace{2ex}57.32\%} & \hspace{2ex}43.61\% \\ \cline{3-10} 
\multicolumn{1}{|c|}{}                         & \multicolumn{1}{l|}{}                         & $acc_{{cert}_r}$    & \multicolumn{1}{l|}{\hspace{3ex}58.28\%} & \multicolumn{1}{l|}{\hspace{3ex}53.23\%} & \hspace{3ex}50.25\% & \multicolumn{1}{l|}{\hspace{3ex}66.94\%} & \hspace{3ex}53.56\% & \multicolumn{1}{l|}{\hspace{3ex}88.03\%} & \hspace{3ex}79.42\% \\ \cline{2-10} 
\multicolumn{1}{|c|}{}                         & \multicolumn{1}{c|}{PL-v2}                         &  $acc_{{cert}_r}$ & \multicolumn{1}{l|}{\hspace{3ex}58.28\%} & \multicolumn{1}{l|}{\hspace{3ex}53.23\%} & \hspace{3ex}50.25\% & \multicolumn{1}{l|}{\hspace{3ex}66.94\%} & \hspace{3ex}53.56\% & \multicolumn{1}{l|}{\hspace{3ex}88.03\%} & \hspace{3ex}79.42\% \\ \cline{2-10} 
\multicolumn{1}{|c|}{}                         & \multicolumn{1}{c|}{\multirow{5}{*}{CC-base}} & $acc_{\textit{clean}}$       & \multicolumn{1}{l|}{80.69\%} & \multicolumn{1}{l|}{80.74\%} & 80.57\% & \multicolumn{1}{l|}{88.64\%} & 88.25\% & \multicolumn{1}{l|}{97.75\%} & 97.81\% \\ \cline{3-10} 
\multicolumn{1}{|c|}{}                         & \multicolumn{1}{l|}{}                         & $acc_{{cert}_d}$   & \multicolumn{1}{l|}{\hspace{1ex}60.78\%} & \multicolumn{1}{l|}{\hspace{1ex}55.62\%} & \hspace{1ex}52.38\% & \multicolumn{1}{l|}{\hspace{1ex}68.47\%} & \hspace{1ex}54.99\% & \multicolumn{1}{l|}{\hspace{1ex}90.10\%} & \hspace{1ex}82.18\% \\ \cline{3-10} 
\multicolumn{1}{|c|}{}                         & \multicolumn{1}{l|}{}                         & $acc_{{cert}_u}$  & \multicolumn{1}{l|}{\hspace{2ex}35.84\%} & \multicolumn{1}{l|}{\hspace{2ex}30.44\%} & \hspace{2ex}26.42\% & \multicolumn{1}{l|}{\hspace{2ex}14.29\%} & \hspace{2ex}8.65\%  & \multicolumn{1}{l|}{\hspace{2ex}57.32\%} & \hspace{2ex}43.61\% \\ \cline{3-10} 
\multicolumn{1}{|c|}{}                         & \multicolumn{1}{l|}{}                         & $acc_{{cert}_r}$    & \multicolumn{1}{l|}{\hspace{3ex}58.28\%} & \multicolumn{1}{l|}{\hspace{3ex}53.23\%} & \hspace{3ex}50.25\% & \multicolumn{1}{l|}{\hspace{3ex}66.94\%} & \hspace{3ex}53.56\% & \multicolumn{1}{l|}{\hspace{3ex}88.03\%} & \hspace{3ex}79.42\% \\ \cline{2-10} 
\multicolumn{1}{|c|}{}                         & \multicolumn{1}{c|}{PL-v1}                        &  $acc_{{cert}_r}$ & \multicolumn{1}{l|}{\hspace{3ex}58.28\%} & \multicolumn{1}{l|}{\hspace{3ex}53.23\%} & \hspace{3ex}50.25\% & \multicolumn{1}{l|}{\hspace{3ex}66.94\%} & \hspace{3ex}53.56\% & \multicolumn{1}{l|}{\hspace{3ex}88.03\%} & \hspace{3ex}79.42\% \\ \hline

\end{tabular}
}
\end{table}
\subsection{Experimental Results and Data Analysis}

\subsubsection{Answering RQ1.} 
Tables \ref{tab:cifar10}, \ref{tab:imagenet}, and \ref{tab:cifar100} summarize the results.
where we use the symbol $\oplus$ to represent the concerned defender to be theoretically infeasible to offer the corresponding certification. 
In each table, the first row shows the patch region size as a percentage of the sample size.
For each patch region size, the three columns represent the clean accuracy, certifiably detectable accuracy, and certifiably unwavering accuracy achieved by the defender on the same row.
The highest number in each column is highlighted.
CC in all these three tables refers to CC-ViT.

\textit{The findings on CIFAR10}: In Table \ref{tab:cifar10}, 
CC is the only defender that can give unwavering certification.
It achieves 78.75\% and 67.73\% in terms of $acc_{{cert}_u}$ for the patch region sizes of 0.4\% and 2.4\%, respectively.
CC, ViP, PC, and PC$_\star$ are competitive with one another in terms of $acc_{\textit{clean}}$ and $acc_{{cert}_d}$ when the patch region size is 0.4\%, where their differences in accuracy are all within 2\%. 
(Note that CC adopts Patchcleanser with Alg. \ref{alg:patchcleanser_revision} for label prediction as $R_1$.)
They generally have higher accuracy in terms of $acc_{{cert}_d}$ than MR$_\star$, PG++$_\star$, SC$_\star$, MR+$_\star$, and MR+.
If the patch region increases to 2.4\%, ViP stands out in terms of $acc_{{cert}_d}$, which is 3.14\% higher than CC and 4.49\% higher than PC.%
\footnote{\label{foot:cifar10}%
We note that ViP and PC use the same criterion to certify a sample as certifiably detectable but with different treatments to realize the same masking strategy on top of the same Visual Transformer architecture.
Furthermore, ViP uses a more advanced pretraining and fine-tuning scheme (i.e., MAE), where the Visual Transformer model after the MAE pre-training and fine-tuning achieves better performance than the one before in the MAE's original experiment \cite{he2022masked}. 
We conjecture that these differences contribute significantly to the differences in accuracy between PC and ViP we observed in our experiment. However, more experiments are needed to confirm or reject the conjecture.
}
CC and PC achieve similar accuracy in terms of $acc_{{cert}_d}$.
Unlike all peer defenders, CC is unique in issuing warnings on the labels returned by a recovery defender ($R_1$) and offers a theoretical guarantee on these recovered labels with suppression of the warnings.

\textit{The findings on ImageNet}: In Table \ref{tab:imagenet}, ViP$_\star$ stands out in terms of $acc_{\textit{clean}}$ and $acc_{{cert}_d}$ for {all three patch region sizes.}
CC,  ViP, and PC achieve similar accuracy in terms of $acc_{\textit{clean}}$ (within 2\% difference).
CC and PC have similar accuracy in terms of $acc_{{cert}_d}$ (around 3.0\%) for the three patch region sizes, but they have lower accuracy than ViP. We have discussed a similar observation in footnote \ref{foot:cifar10} when reporting the findings on CIFAR10 above. 
PC$_\star$ has higher accuracy than CC in all columns. 
MR+, PG++$_\star$, SC$_\star$ and MR+$_\star$ achieve lower accuracy in terms of $acc_{\textit{clean}}$ and $acc_{{cert}_d}$ than CC. Like Table \ref{tab:cifar10}, only CC can provide unwavering certification, achieving 44.68\%, 39.52\%, and 35.48\% in terms of $acc_{{cert}_u}$ for the patch region sizes of 1\%, 2\%, and 3\%, respectively.

\textit{The findings on CIFAR100}: In Table \ref{tab:cifar100}, CC achieves similar accuracy to PC in terms of $acc_{\textit{clean}}$ (within 0.2\% difference) for the patch region sizes of 0.4\% and 2.4\%, which are higher than ViP by 2.90\% {on} average. 
ViP achieves the highest $acc_{{cert}_d}$ for both patch region sizes. 
CC is also more accurate in terms of $acc_{{cert}_d}$ than MR+.
Only CC can provide unwavering certification with 52.50\% and 39.07\% in terms of $acc_{{cert}_{u}}$ for the patch region sizes of 0.4\% and 2.4\%, respectively.

Overall, with the additional constraint for providing unwavering certification, in terms of $acc_{\textit{clean}}$ and $acc_{{cert}_d}$, CC incurs lower accuracy than ViP and is competitive with  PC in performance on all three datasets. Furthermore, it gains non-trivial certifiably unwavering accuracy.

\subsubsection{Answering RQ2.} Table \ref{tab:cc-base} shows the accuracy results for CC and CC-base and their primary recovery models (i.e., $R_1$), which are denoted by \textbf{PL-v2} and \textbf{PL-v1}, respectively, standing for PatchCleanser with Alg. \ref{alg:patchcleanser_revision} and PatchCleanser with Alg. \ref{alg:patchcleanser_original} (i.e., the original PatchCleanser) for label prediction, respectively, where $R_1$ is configured with the base model (ViT or ResNet).
Table \ref{tab:cc-base} shows the datasets and patch region sizes in the first two rows. 
Each other row shows the accuracy (indicated by the third column) achieved by a defender (indicated in the second column) using the base model (indicated in the first column) to configure the recovery defender (PL-v2, PL-v1, and $R_1$ for CC and CC-base).
CC and CC-base directly adopt PL-v2 and PL-v1 as $R_1$, respectively.  
Their clean accuracy and certifiably recoverable accuracy are the same as those for PL-v2 and PL-v1, respectively. 
We report their certifiably recoverable accuracy for validation.

We have several observations from Table \ref{tab:cc-base}.
First, using ViT as the base model makes CC-ViT and CC-base-ViT consistently outperform their counterparts that use ResNet as the base model (CC-RN and CC-base-RN) in each column.
This observation is consistent with the findings in \cite{patchcensor} for detection defenders and those in \cite{salman2022certified} for recovery defenders that ViT outperforms ResNet for patch robustness certification by serving as the base models of defenders.
Second, using ViT and ResNet as the base models, 
 CC-ViT and CC-RN also consistently outperform CC-base-ViT and CC-base-RN in terms of $acc_{cert_{d}}$ in each column, respectively, which shows that 
 the deeper analysis on the recovery semantics made by \textit{CrossCert} pays off.
 Also, the certifiably detectable accuracy for using each of CC-ViT, CC-RN, CC-base-ViT, and CC-base-RN is consistently higher than
 the corresponding certifiably recoverable accuracy and certifiably unwavering accuracy of the same defender in the same column.
 Third, as expected (and discussed in Section~\ref{sec:analysis-masking-recovery}), for each of ViT and ResNet as the base model,
 the certifiably recoverable accuracy for CC and that for CC-base are consistently the same as these for PL-v2 and PL-v1, and 
 the certifiably unwavering accuracy for CC and that for CC-base are also consistently the same.
 Last but not least, 
an interesting observation is that, for using each of ViT and ResNet as the base model, CC and CC-base share the same clean accuracy in each column, indicating that their underlying PL-v2 and PL-v1 also have the same clean accuracy.
We conjecture that Alg \ref{alg:patchcleanser_original} and Alg. \ref{alg:patchcleanser_revision} always return the same label for the same sample.
We leave the theoretical investigation for this conjecture as future work.
%

\subsection{Threats to Validity}
We evaluate three top-performing peer defenders and CC in the experiments and follow common practice in the latest previous studies \cite{patchcensor,li2022vip} on certified detection to compare the {reported results (in numbers)} of more defenders reported in the literature on certified detection due to our limited GPU resources. 
We evaluate the defenders in four metrics
$acc_{\textit{clean}}$, $acc_{{cert}_d}$, $acc_{{cert}_r}$, and $acc_{{cert}_u}$.
The first three are popularly measured in previous experiments \cite{patchcensor,li2022vip, xiang2022patchcleanser}, and the fourth one is original. 
Some previous work (MR \cite{mccoyd2020minority}, PG++ \cite{xiang2021patchguard++}, and SC \cite{han2021scalecert}) reported their results on metrics with the same names as our metrics but excluding those benign samples with warnings produced from a defender when calculating the corresponding accuracy.
We follow \cite{patchcensor,li2022vip} to replicate their reported results into our tables with the same metric names for comparison and remind readers in this section that they have different definitions.
{Another interpretation of 
$acc_{{cert}_u}$ is that it measures
the extent of \textit{true negatives} of certified detection guaranteed by a defender. 
{We believe that using an independent metric on true negatives can be more comprehensive in evaluating a defender.}
Lowering the certified accuracy for detection to that for recovery is undesirable as such detectors cannot achieve state-of-the-art certified detection performance.}
We are aware that conducting experiments with more peer defenders configured with more base models on more datasets and on more metrics can further strengthen the validity of the experiment.
On CIFAR100, CC with ResNet achieves a low performance. We believe that the hyperparameters we used for it, as stated in Section \ref{sec:implement}, are a threat on the performance. 
The time costs incurred by CC to certify samples are around 9 hours for CIFAR10 and CIFAR100 and 41 hours for ImageNet, 
in which $R_1$, $R_2$, and the \textit{CrossCert} framework account for 
83.3\%, 16.5\%, and 0.2\% on average, respectively.
Like \cite{patchcensor,han2021scalecert,xiang2021patchguard++},
we only conduct the experiments in single-patch situations.
We also only conduct the experiments on three datasets with limited patch sizes similar to \cite{han2021scalecert,levine2020randomized,xiang2021patchguard++,chiang2020certified}. We have generalized \textit{CrossCert} to two popular backbones and leave the further generalization as future work.
Similar to the evaluations of existing robustness certification defenders \cite{li2022vip,patchcensor,xiang2021patchguard++,xiang2022patchcleanser,salman2022certified,xiang2021patchguard}, we have not evaluated \textit{CrossCert} on real adversarial attacks.
The implementation of all techniques 
may contain bugs. 
We have conducted testing and resolved bugs that were identified.

\section{Related work}\label{sec:rel}
Extensive research has been done to defend against adversarial patch attacks. Various empirical defenses have been proposed \cite{naseer2019local,hayes2018visible}, but their defenses can easily be compromised if an attacker is aware of their defense strategies \cite{chiang2020certified}. Consequently, to construct a defense against such situations, the class of robustness certification techniques against patch attacks emerges \cite{mccoyd2020minority,xiang2021patchguard++,patchcensor,han2021scalecert,li2022vip,chiang2020certified,levine2020randomized,salman2022certified,chen2022towards,lin2021certified,xiang2021patchguard,metzen2021efficient,zhang2020clipped,xiang2022patchcleanser}.  Certified recovery \cite{chiang2020certified,levine2020randomized,salman2022certified,chen2022towards,lin2021certified,xiang2021patchguard,metzen2021efficient,zhang2020clipped,xiang2022patchcleanser} aims higher by recovering the labels of benign samples in the presence of the adversarial patch, while certified detection \cite{mccoyd2020minority,xiang2021patchguard++,patchcensor,han2021scalecert,li2022vip} focuses on providing warnings against such a patch \cite{xiang2021patchguard,xiangshort}.

Chiang et al. \cite{chiang2020certified} propose the first work on certified recovery technique against patch attacks with IBP \cite{gowal2019scalable}. Subsequently, several voting-based techniques (DRS \cite{levine2020randomized}, ViT \cite{salman2022certified}, ViP \cite{li2022vip}, and ECViT \cite{chen2022towards}) utilize the margin of majority votes among the smoothing ablations to provide certification guarantees. 
Vision Transformer-based techniques \cite{chen2022towards, salman2022certified,li2022vip} are empirically found to achieve higher certified accuracy than convolutional neural network-based techniques.
PG \cite{xiang2021patchguard}, MC \cite{zhou2023majority}, and BagCert \cite{metzen2021efficient} explore the constraints among the ablations to eliminate certain infeasible cases to theoretically enhance the certified accuracy.  
Xiang et al. \cite{xiang2022patchcleanser} propose the first work (PatchCleanser) in the class of masking-based certified recovery to recover benign labels. Unlike them, \textit{CrossCert} contributes to differentiate warning outputs by recovery semantics through Alg. \ref{alg:patchcleanser_revision}.

A stringent condition imposed by a recovery guarantee usually results in a more limited certifiably recoverable accuracy compared to clean accuracy, making recovery defenders challenging for uses in the real world \cite{patchcensor} and motivating subsequent research to study certified detection. 
MR \cite{mccoyd2020minority} firstly introduces a mask-sliding mechanism for detection certification, which achieves a relatively low certifiably detectable accuracy and limited scalability. 
PG++ \cite{xiang2021patchguard++} and SC \cite{han2021scalecert} introduce feature masking and neural network compression strategies, improving their scalability compared to MR. PC and ViP \cite{patchcensor, li2022vip} further refine the masking mechanism and adapt ViT as their backbones, resulting in higher performances. Unlike \textit{CrossCert}, none of them explores the possibility of cross-checking between recovery defenders to provide detection certification. They provide \emph{no} unwavering certification,
incurring a greater chance of invoking a fall-back strategy to resolve warnings, lowering their support for automation.

\section{Conclusion}\label{sec:con}
We have presented a novel cross-check detection framework \textit{CrossCert} for a novel certification known as the unwavering certification.
We have proven its correctness and conducted experiments to show its high effectiveness in providing unwavering certification and detection certification.
We leave the improvement of its efficiency, the generalization of building it atop a set of defenders, and the generalization of it for robustness certification across different adversaries as future work.

\section*{Data availability}
Our replication package \cite{CrossCert_github} is available on GitHub. ImageNet can be downloaded at \url{image-net.org}, and other datasets will automatically be downloaded by the script. 

\clearpage
\bibliographystyle{ACM-Reference-Format}

\begin{thebibliography}{44}


\ifx \showCODEN    \undefined \def \showCODEN     #1{\unskip}     \fi
\ifx \showDOI      \undefined \def \showDOI       #1{#1}\fi
\ifx \showISBNx    \undefined \def \showISBNx     #1{\unskip}     \fi
\ifx \showISBNxiii \undefined \def \showISBNxiii  #1{\unskip}     \fi
\ifx \showISSN     \undefined \def \showISSN      #1{\unskip}     \fi
\ifx \showLCCN     \undefined \def \showLCCN      #1{\unskip}     \fi
\ifx \shownote     \undefined \def \shownote      #1{#1}          \fi
\ifx \showarticletitle \undefined \def \showarticletitle #1{#1}   \fi
\ifx \showURL      \undefined \def \showURL       {\relax}        \fi
\providecommand\bibfield[2]{#2}
\providecommand\bibinfo[2]{#2}
\providecommand\natexlab[1]{#1}
\providecommand\showeprint[2][]{arXiv:#2}

\bibitem[Cro(2023)]%
        {CrossCert_github}
 \bibinfo{year}{2023}\natexlab{}.
\newblock \bibinfo{title}{CrossCert}.
\newblock
\newblock
\urldef\tempurl%
\url{https://github.com/kio-cs/CrossCert}
\showURL{%
\tempurl}


\bibitem[Brown et~al\mbox{.}(2017)]%
        {brown2017adversarial}
\bibfield{author}{\bibinfo{person}{Tom~B Brown}, \bibinfo{person}{Dandelion Man{\'e}}, \bibinfo{person}{Aurko Roy}, \bibinfo{person}{Mart{\'\i}n Abadi}, {and} \bibinfo{person}{Justin Gilmer}.} \bibinfo{year}{2017}\natexlab{}.
\newblock \showarticletitle{Adversarial patch}.
\newblock \bibinfo{journal}{\emph{arXiv preprint arXiv:1712.09665}} (\bibinfo{year}{2017}).
\newblock


\bibitem[Chen et~al\mbox{.}(2022)]%
        {chen2022towards}
\bibfield{author}{\bibinfo{person}{Zhaoyu Chen}, \bibinfo{person}{Bo Li}, \bibinfo{person}{Jianghe Xu}, \bibinfo{person}{Shuang Wu}, \bibinfo{person}{Shouhong Ding}, {and} \bibinfo{person}{Wenqiang Zhang}.} \bibinfo{year}{2022}\natexlab{}.
\newblock \showarticletitle{Towards practical certifiable patch defense with vision transformer}. In \bibinfo{booktitle}{\emph{Proceedings of the IEEE/CVF Conference on Computer Vision and Pattern Recognition}}. \bibinfo{pages}{15148--15158}.
\newblock


\bibitem[Deng et~al\mbox{.}(2009)]%
        {deng2009imagenet}
\bibfield{author}{\bibinfo{person}{Jia Deng}, \bibinfo{person}{Wei Dong}, \bibinfo{person}{Richard Socher}, \bibinfo{person}{Li-Jia Li}, \bibinfo{person}{Kai Li}, {and} \bibinfo{person}{Li Fei-Fei}.} \bibinfo{year}{2009}\natexlab{}.
\newblock \showarticletitle{Imagenet: A large-scale hierarchical image database}. In \bibinfo{booktitle}{\emph{2009 IEEE conference on computer vision and pattern recognition}}. Ieee, \bibinfo{pages}{248--255}.
\newblock


\bibitem[Dosovitskiy et~al\mbox{.}(2021)]%
        {dosovitskiy2021an}
\bibfield{author}{\bibinfo{person}{Alexey Dosovitskiy}, \bibinfo{person}{Lucas Beyer}, \bibinfo{person}{Alexander Kolesnikov}, \bibinfo{person}{Dirk Weissenborn}, \bibinfo{person}{Xiaohua Zhai}, \bibinfo{person}{Thomas Unterthiner}, \bibinfo{person}{Mostafa Dehghani}, \bibinfo{person}{Matthias Minderer}, \bibinfo{person}{Georg Heigold}, \bibinfo{person}{Sylvain Gelly}, \bibinfo{person}{Jakob Uszkoreit}, {and} \bibinfo{person}{Neil Houlsby}.} \bibinfo{year}{2021}\natexlab{}.
\newblock \showarticletitle{An Image is Worth 16x16 Words: Transformers for Image Recognition at Scale}. In \bibinfo{booktitle}{\emph{International Conference on Learning Representations}}.
\newblock
\urldef\tempurl%
\url{https://openreview.net/forum?id=YicbFdNTTy}
\showURL{%
\tempurl}


\bibitem[Eykholt et~al\mbox{.}(2018)]%
        {eykholt2018robust}
\bibfield{author}{\bibinfo{person}{Kevin Eykholt}, \bibinfo{person}{Ivan Evtimov}, \bibinfo{person}{Earlence Fernandes}, \bibinfo{person}{Bo Li}, \bibinfo{person}{Amir Rahmati}, \bibinfo{person}{Chaowei Xiao}, \bibinfo{person}{Atul Prakash}, \bibinfo{person}{Tadayoshi Kohno}, {and} \bibinfo{person}{Dawn Song}.} \bibinfo{year}{2018}\natexlab{}.
\newblock \showarticletitle{Robust physical-world attacks on deep learning visual classification}. In \bibinfo{booktitle}{\emph{Proceedings of the IEEE conference on computer vision and pattern recognition}}. \bibinfo{pages}{1625--1634}.
\newblock


\bibitem[facebookresearch(2022)]%
        {MAE_code}
\bibfield{author}{\bibinfo{person}{facebookresearch}.} \bibinfo{year}{2022}\natexlab{}.
\newblock \bibinfo{title}{Official implementation for MAE}.
\newblock \bibinfo{howpublished}{\url{https://github.com/facebookresearch/mae}}.
\newblock


\bibitem[Gowal et~al\mbox{.}(2019)]%
        {gowal2019scalable}
\bibfield{author}{\bibinfo{person}{Sven Gowal}, \bibinfo{person}{Krishnamurthy~Dj Dvijotham}, \bibinfo{person}{Robert Stanforth}, \bibinfo{person}{Rudy Bunel}, \bibinfo{person}{Chongli Qin}, \bibinfo{person}{Jonathan Uesato}, \bibinfo{person}{Relja Arandjelovic}, \bibinfo{person}{Timothy Mann}, {and} \bibinfo{person}{Pushmeet Kohli}.} \bibinfo{year}{2019}\natexlab{}.
\newblock \showarticletitle{Scalable verified training for provably robust image classification}. In \bibinfo{booktitle}{\emph{Proceedings of the IEEE/CVF International Conference on Computer Vision}}. \bibinfo{pages}{4842--4851}.
\newblock


\bibitem[Gu et~al\mbox{.}(2019)]%
        {gu2019badnets}
\bibfield{author}{\bibinfo{person}{Tianyu Gu}, \bibinfo{person}{Kang Liu}, \bibinfo{person}{Brendan Dolan-Gavitt}, {and} \bibinfo{person}{Siddharth Garg}.} \bibinfo{year}{2019}\natexlab{}.
\newblock \showarticletitle{Badnets: Evaluating backdooring attacks on deep neural networks}.
\newblock \bibinfo{journal}{\emph{IEEE Access}}  \bibinfo{volume}{7} (\bibinfo{year}{2019}), \bibinfo{pages}{47230--47244}.
\newblock


\bibitem[Han et~al\mbox{.}(2021)]%
        {han2021scalecert}
\bibfield{author}{\bibinfo{person}{Husheng Han}, \bibinfo{person}{Kaidi Xu}, \bibinfo{person}{Xing Hu}, \bibinfo{person}{Xiaobing Chen}, \bibinfo{person}{Ling Liang}, \bibinfo{person}{Zidong Du}, \bibinfo{person}{Qi Guo}, \bibinfo{person}{Yanzhi Wang}, {and} \bibinfo{person}{Yunji Chen}.} \bibinfo{year}{2021}\natexlab{}.
\newblock \showarticletitle{Scalecert: Scalable certified defense against adversarial patches with sparse superficial layers}.
\newblock \bibinfo{journal}{\emph{Advances in Neural Information Processing Systems}}  \bibinfo{volume}{34} (\bibinfo{year}{2021}), \bibinfo{pages}{28169--28181}.
\newblock


\bibitem[Hayes(2018)]%
        {hayes2018visible}
\bibfield{author}{\bibinfo{person}{Jamie Hayes}.} \bibinfo{year}{2018}\natexlab{}.
\newblock \showarticletitle{On visible adversarial perturbations \& digital watermarking}. In \bibinfo{booktitle}{\emph{Proceedings of the IEEE Conference on Computer Vision and Pattern Recognition Workshops}}. \bibinfo{pages}{1597--1604}.
\newblock


\bibitem[He et~al\mbox{.}(2022)]%
        {he2022masked}
\bibfield{author}{\bibinfo{person}{Kaiming He}, \bibinfo{person}{Xinlei Chen}, \bibinfo{person}{Saining Xie}, \bibinfo{person}{Yanghao Li}, \bibinfo{person}{Piotr Doll{\'a}r}, {and} \bibinfo{person}{Ross Girshick}.} \bibinfo{year}{2022}\natexlab{}.
\newblock \showarticletitle{Masked autoencoders are scalable vision learners}. In \bibinfo{booktitle}{\emph{Proceedings of the IEEE/CVF conference on computer vision and pattern recognition}}. \bibinfo{pages}{16000--16009}.
\newblock


\bibitem[He et~al\mbox{.}(2016)]%
        {he2016deep}
\bibfield{author}{\bibinfo{person}{Kaiming He}, \bibinfo{person}{Xiangyu Zhang}, \bibinfo{person}{Shaoqing Ren}, {and} \bibinfo{person}{Jian Sun}.} \bibinfo{year}{2016}\natexlab{}.
\newblock \showarticletitle{Deep residual learning for image recognition}. In \bibinfo{booktitle}{\emph{Proceedings of the IEEE conference on computer vision and pattern recognition}}. \bibinfo{pages}{770--778}.
\newblock


\bibitem[Huang(2022)]%
        {PatchCensor_code}
\bibfield{author}{\bibinfo{person}{Yuheng Huang}.} \bibinfo{year}{2022}\natexlab{}.
\newblock \bibinfo{title}{Official implementation for PatchCensor}.
\newblock \bibinfo{howpublished}{\url{https://github.com/YuhengHuang42/PatchCensor}}.
\newblock


\bibitem[Huang et~al\mbox{.}(2023)]%
        {patchcensor}
\bibfield{author}{\bibinfo{person}{Yuheng Huang}, \bibinfo{person}{Lei Ma}, {and} \bibinfo{person}{Yuanchun Li}.} \bibinfo{year}{2023}\natexlab{}.
\newblock \showarticletitle{PatchCensor: Patch Robustness Certification for Transformers via Exhaustive Testing}.
\newblock \bibinfo{journal}{\emph{ACM Trans. Softw. Eng. Methodol.}} (\bibinfo{date}{apr} \bibinfo{year}{2023}).
\newblock
\showISSN{1049-331X}
\urldef\tempurl%
\url{https://doi.org/10.1145/3591870}
\showDOI{\tempurl}
\newblock
\shownote{Just Accepted}.


\bibitem[huggingface(2023a)]%
        {resnet-50}
\bibfield{author}{\bibinfo{person}{huggingface}.} \bibinfo{year}{2023}\natexlab{a}.
\newblock \bibinfo{title}{Source of ResNet-50}.
\newblock \bibinfo{howpublished}{\url{https://huggingface.co/timm/resnetv2_50x1_bit.goog_distilled_in1k/tree/main}}.
\newblock


\bibitem[huggingface(2023b)]%
        {vit-b16-224}
\bibfield{author}{\bibinfo{person}{huggingface}.} \bibinfo{year}{2023}\natexlab{b}.
\newblock \bibinfo{title}{Source of vit-b16-224}.
\newblock \bibinfo{howpublished}{\url{https://huggingface.co/timm/vit_base_patch16_224.augreg2_in21k_ft_in1k}}.
\newblock


\bibitem[Krizhevsky et~al\mbox{.}(2009)]%
        {krizhevsky2009learning}
\bibfield{author}{\bibinfo{person}{Alex Krizhevsky}, \bibinfo{person}{Geoffrey Hinton}, {et~al\mbox{.}}} \bibinfo{year}{2009}\natexlab{}.
\newblock \showarticletitle{Learning multiple layers of features from tiny images}.
\newblock  (\bibinfo{year}{2009}).
\newblock


\bibitem[Kuutti et~al\mbox{.}(2021)]%
        {2021SurveyK}
\bibfield{author}{\bibinfo{person}{Sampo Kuutti}, \bibinfo{person}{Richard Bowden}, \bibinfo{person}{Yaochu Jin}, \bibinfo{person}{Phil Barber}, {and} \bibinfo{person}{Saber Fallah}.} \bibinfo{year}{2021}\natexlab{}.
\newblock \showarticletitle{A Survey of Deep Learning Applications to Autonomous Vehicle Control}.
\newblock \bibinfo{journal}{\emph{IEEE Transactions on Intelligent Transportation Systems}} \bibinfo{volume}{22}, \bibinfo{number}{2} (\bibinfo{year}{2021}), \bibinfo{pages}{712--733}.
\newblock
\urldef\tempurl%
\url{https://doi.org/10.1109/TITS.2019.2962338}
\showDOI{\tempurl}


\bibitem[Levine and Feizi(2020)]%
        {levine2020randomized}
\bibfield{author}{\bibinfo{person}{Alexander Levine} {and} \bibinfo{person}{Soheil Feizi}.} \bibinfo{year}{2020}\natexlab{}.
\newblock \showarticletitle{(De) randomized smoothing for certifiable defense against patch attacks}. In \bibinfo{booktitle}{\emph{Proceedings of the 34th International Conference on Neural Information Processing Systems}}. \bibinfo{pages}{6465--6475}.
\newblock


\bibitem[Li(2022)]%
        {ViP_code}
\bibfield{author}{\bibinfo{person}{Junbo Li}.} \bibinfo{year}{2022}\natexlab{}.
\newblock \bibinfo{title}{Official implementation for ViP}.
\newblock \bibinfo{howpublished}{\url{https://github.com/UCSC-VLAA/vit_cert}}.
\newblock


\bibitem[Li et~al\mbox{.}(2022)]%
        {li2022vip}
\bibfield{author}{\bibinfo{person}{Junbo Li}, \bibinfo{person}{Huan Zhang}, {and} \bibinfo{person}{Cihang Xie}.} \bibinfo{year}{2022}\natexlab{}.
\newblock \showarticletitle{ViP: Unified Certified Detection and Recovery for Patch Attack with Vision Transformers}. In \bibinfo{booktitle}{\emph{Computer Vision--ECCV 2022: 17th European Conference, Tel Aviv, Israel, October 23--27, 2022, Proceedings, Part XXV}}. Springer, \bibinfo{pages}{573--587}.
\newblock


\bibitem[Lin et~al\mbox{.}(2021)]%
        {lin2021certified}
\bibfield{author}{\bibinfo{person}{Wan-Yi Lin}, \bibinfo{person}{Fatemeh Sheikholeslami}, \bibinfo{person}{Leslie Rice}, \bibinfo{person}{J~Zico Kolter}, {et~al\mbox{.}}} \bibinfo{year}{2021}\natexlab{}.
\newblock \showarticletitle{Certified robustness against adversarial patch attacks via randomized cropping}. In \bibinfo{booktitle}{\emph{ICML 2021 Workshop on Adversarial Machine Learning}}.
\newblock


\bibitem[Litjens et~al\mbox{.}(2017)]%
        {litjens2017survey}
\bibfield{author}{\bibinfo{person}{Geert Litjens}, \bibinfo{person}{Thijs Kooi}, \bibinfo{person}{Babak~Ehteshami Bejnordi}, \bibinfo{person}{Arnaud Arindra~Adiyoso Setio}, \bibinfo{person}{Francesco Ciompi}, \bibinfo{person}{Mohsen Ghafoorian}, \bibinfo{person}{Jeroen~Awm Van Der~Laak}, \bibinfo{person}{Bram Van~Ginneken}, {and} \bibinfo{person}{Clara~I S{\'a}nchez}.} \bibinfo{year}{2017}\natexlab{}.
\newblock \showarticletitle{A survey on deep learning in medical image analysis}.
\newblock \bibinfo{journal}{\emph{Medical image analysis}}  \bibinfo{volume}{42} (\bibinfo{year}{2017}), \bibinfo{pages}{60--88}.
\newblock


\bibitem[Ma et~al\mbox{.}(2019)]%
        {ma2019nic}
\bibfield{author}{\bibinfo{person}{Shiqing Ma}, \bibinfo{person}{Yingqi Liu}, \bibinfo{person}{Guanhong Tao}, \bibinfo{person}{Wen-Chuan Lee}, {and} \bibinfo{person}{Xiangyu Zhang}.} \bibinfo{year}{2019}\natexlab{}.
\newblock \showarticletitle{Nic: Detecting adversarial samples with neural network invariant checking}. In \bibinfo{booktitle}{\emph{26th Annual Network And Distributed System Security Symposium (NDSS 2019)}}. Internet Soc.
\newblock


\bibitem[maintainers and contributors(2016)]%
        {torchvision2016}
\bibfield{author}{\bibinfo{person}{TorchVision maintainers} {and} \bibinfo{person}{contributors}.} \bibinfo{year}{2016}\natexlab{}.
\newblock \bibinfo{booktitle}{\emph{TorchVision: PyTorch's Computer Vision library}}.
\newblock


\bibitem[McCoyd et~al\mbox{.}(2020)]%
        {mccoyd2020minority}
\bibfield{author}{\bibinfo{person}{Michael McCoyd}, \bibinfo{person}{Won Park}, \bibinfo{person}{Steven Chen}, \bibinfo{person}{Neil Shah}, \bibinfo{person}{Ryan Roggenkemper}, \bibinfo{person}{Minjune Hwang}, \bibinfo{person}{Jason~Xinyu Liu}, {and} \bibinfo{person}{David Wagner}.} \bibinfo{year}{2020}\natexlab{}.
\newblock \showarticletitle{Minority reports defense: Defending against adversarial patches}. In \bibinfo{booktitle}{\emph{Applied Cryptography and Network Security Workshops: ACNS 2020 Satellite Workshops, AIBlock, AIHWS, AIoTS, Cloud S\&P, SCI, SecMT, and SiMLA, Rome, Italy, October 19--22, 2020, Proceedings}}. Springer, \bibinfo{pages}{564--582}.
\newblock


\bibitem[Metzen and Yatsura(2021)]%
        {metzen2021efficient}
\bibfield{author}{\bibinfo{person}{Jan~Hendrik Metzen} {and} \bibinfo{person}{Maksym Yatsura}.} \bibinfo{year}{2021}\natexlab{}.
\newblock \showarticletitle{Efficient Certified Defenses Against Patch Attacks on Image Classifiers}. In \bibinfo{booktitle}{\emph{International Conference on Learning Representations}}.
\newblock
\urldef\tempurl%
\url{https://openreview.net/forum?id=hr-3PMvDpil}
\showURL{%
\tempurl}


\bibitem[Naseer et~al\mbox{.}(2019)]%
        {naseer2019local}
\bibfield{author}{\bibinfo{person}{Muzammal Naseer}, \bibinfo{person}{Salman Khan}, {and} \bibinfo{person}{Fatih Porikli}.} \bibinfo{year}{2019}\natexlab{}.
\newblock \showarticletitle{Local gradients smoothing: Defense against localized adversarial attacks}. In \bibinfo{booktitle}{\emph{2019 IEEE Winter Conference on Applications of Computer Vision (WACV)}}. IEEE, \bibinfo{pages}{1300--1307}.
\newblock


\bibitem[Saha et~al\mbox{.}(2023)]%
        {saha2023revisiting}
\bibfield{author}{\bibinfo{person}{Aniruddha Saha}, \bibinfo{person}{Shuhua Yu}, \bibinfo{person}{Mohammad~Sadegh Norouzzadeh}, \bibinfo{person}{Wan-Yi Lin}, {and} \bibinfo{person}{Chaithanya~Kumar Mummadi}.} \bibinfo{year}{2023}\natexlab{}.
\newblock \showarticletitle{Revisiting Image Classifier Training for Improved Certified Robust Defense against Adversarial Patches}.
\newblock \bibinfo{journal}{\emph{Transactions on Machine Learning Research}} (\bibinfo{year}{2023}).
\newblock
\showISSN{2835-8856}
\urldef\tempurl%
\url{https://openreview.net/forum?id=2tdhQMLg36}
\showURL{%
\tempurl}


\bibitem[Salman et~al\mbox{.}(2022)]%
        {salman2022certified}
\bibfield{author}{\bibinfo{person}{Hadi Salman}, \bibinfo{person}{Saachi Jain}, \bibinfo{person}{Eric Wong}, {and} \bibinfo{person}{Aleksander Madry}.} \bibinfo{year}{2022}\natexlab{}.
\newblock \showarticletitle{Certified patch robustness via smoothed vision transformers}. In \bibinfo{booktitle}{\emph{Proceedings of the IEEE/CVF Conference on Computer Vision and Pattern Recognition}}. \bibinfo{pages}{15137--15147}.
\newblock


\bibitem[Szegedy et~al\mbox{.}(2013)]%
        {szegedy2013intriguing}
\bibfield{author}{\bibinfo{person}{Christian Szegedy}, \bibinfo{person}{Wojciech Zaremba}, \bibinfo{person}{Ilya Sutskever}, \bibinfo{person}{Joan Bruna}, \bibinfo{person}{Dumitru Erhan}, \bibinfo{person}{Ian Goodfellow}, {and} \bibinfo{person}{Rob Fergus}.} \bibinfo{year}{2013}\natexlab{}.
\newblock \showarticletitle{Intriguing properties of neural networks}.
\newblock \bibinfo{journal}{\emph{arXiv preprint arXiv:1312.6199}} (\bibinfo{year}{2013}).
\newblock


\bibitem[Usman et~al\mbox{.}(2021)]%
        {usman2021nn}
\bibfield{author}{\bibinfo{person}{Muhammad Usman}, \bibinfo{person}{Divya Gopinath}, \bibinfo{person}{Youcheng Sun}, \bibinfo{person}{Yannic Noller}, {and} \bibinfo{person}{Corina~S P{\u{a}}s{\u{a}}reanu}.} \bibinfo{year}{2021}\natexlab{}.
\newblock \showarticletitle{Nn repair: Constraint-based repair of neural network classifiers}. In \bibinfo{booktitle}{\emph{Computer Aided Verification: 33rd International Conference, CAV 2021, Virtual Event, July 20--23, 2021, Proceedings, Part I 33}}. Springer, \bibinfo{pages}{3--25}.
\newblock


\bibitem[Wightman(2019)]%
        {rw2019timm}
\bibfield{author}{\bibinfo{person}{Ross Wightman}.} \bibinfo{year}{2019}\natexlab{}.
\newblock \bibinfo{title}{PyTorch Image Models}.
\newblock \bibinfo{howpublished}{\url{https://github.com/rwightman/pytorch-image-models}}.
\newblock
\urldef\tempurl%
\url{https://doi.org/10.5281/zenodo.4414861}
\showDOI{\tempurl}


\bibitem[Xiang et~al\mbox{.}(2021)]%
        {xiang2021patchguard}
\bibfield{author}{\bibinfo{person}{Chong Xiang}, \bibinfo{person}{Arjun~Nitin Bhagoji}, \bibinfo{person}{Vikash Sehwag}, {and} \bibinfo{person}{Prateek Mittal}.} \bibinfo{year}{2021}\natexlab{}.
\newblock \showarticletitle{PatchGuard: A Provably Robust Defense against Adversarial Patches via Small Receptive Fields and Masking}. In \bibinfo{booktitle}{\emph{30th USENIX Security Symposium (USENIX Security 21)}}. \bibinfo{pages}{2237--2254}.
\newblock


\bibitem[Xiang et~al\mbox{.}(2022)]%
        {xiang2022patchcleanser}
\bibfield{author}{\bibinfo{person}{Chong Xiang}, \bibinfo{person}{Saeed Mahloujifar}, {and} \bibinfo{person}{Prateek Mittal}.} \bibinfo{year}{2022}\natexlab{}.
\newblock \showarticletitle{PatchCleanser: Certifiably Robust Defense against Adversarial Patches for Any Image Classifier}. In \bibinfo{booktitle}{\emph{31st USENIX Security Symposium (USENIX Security 22)}}. \bibinfo{pages}{2065--2082}.
\newblock


\bibitem[Xiang and Mittal(2021)]%
        {xiang2021patchguard++}
\bibfield{author}{\bibinfo{person}{Chong Xiang} {and} \bibinfo{person}{Prateek Mittal}.} \bibinfo{year}{2021}\natexlab{}.
\newblock \showarticletitle{Patchguard++: Efficient provable attack detection against adversarial patches}.
\newblock \bibinfo{journal}{\emph{arXiv preprint arXiv:2104.12609}} (\bibinfo{year}{2021}).
\newblock


\bibitem[Xiang et~al\mbox{.}(2023)]%
        {xiangshort}
\bibfield{author}{\bibinfo{person}{Chong Xiang}, \bibinfo{person}{Chawin Sitawarin}, \bibinfo{person}{Tong Wu}, {and} \bibinfo{person}{Prateek Mittal}.} \bibinfo{year}{2023}\natexlab{}.
\newblock \showarticletitle{Short: Certifiably Robust Perception Against Adversarial Patch Attacks: A Survey}.
\newblock \bibinfo{journal}{\emph{Proceedings Inaugural International Symposium on Vehicle Security \& Privacy}} (\bibinfo{year}{2023}).
\newblock
\urldef\tempurl%
\url{https://api.semanticscholar.org/CorpusID:257490567}
\showURL{%
\tempurl}


\bibitem[yeh Chiang* et~al\mbox{.}(2020)]%
        {chiang2020certified}
\bibfield{author}{\bibinfo{person}{Ping yeh Chiang*}, \bibinfo{person}{Renkun Ni*}, \bibinfo{person}{Ahmed Abdelkader}, \bibinfo{person}{Chen Zhu}, \bibinfo{person}{Christoph Studor}, {and} \bibinfo{person}{Tom Goldstein}.} \bibinfo{year}{2020}\natexlab{}.
\newblock \showarticletitle{Certified Defenses for Adversarial Patches}. In \bibinfo{booktitle}{\emph{International Conference on Learning Representations}}.
\newblock
\urldef\tempurl%
\url{https://openreview.net/forum?id=HyeaSkrYPH}
\showURL{%
\tempurl}


\bibitem[Zhang et~al\mbox{.}(2022)]%
        {zhang2022towards}
\bibfield{author}{\bibinfo{person}{Huangzhao Zhang}, \bibinfo{person}{Zhiyi Fu}, \bibinfo{person}{Ge Li}, \bibinfo{person}{Lei Ma}, \bibinfo{person}{Zhehao Zhao}, \bibinfo{person}{Hua’an Yang}, \bibinfo{person}{Yizhe Sun}, \bibinfo{person}{Yang Liu}, {and} \bibinfo{person}{Zhi Jin}.} \bibinfo{year}{2022}\natexlab{}.
\newblock \showarticletitle{Towards robustness of deep program processing models—detection, estimation, and enhancement}.
\newblock \bibinfo{journal}{\emph{ACM Transactions on Software Engineering and Methodology (TOSEM)}} \bibinfo{volume}{31}, \bibinfo{number}{3} (\bibinfo{year}{2022}), \bibinfo{pages}{1--40}.
\newblock


\bibitem[Zhang et~al\mbox{.}(2021)]%
        {zhang2021out}
\bibfield{author}{\bibinfo{person}{Zhen Zhang}, \bibinfo{person}{Peng Wu}, \bibinfo{person}{Yuhang Chen}, {and} \bibinfo{person}{Jing Su}.} \bibinfo{year}{2021}\natexlab{}.
\newblock \showarticletitle{Out-of-Distribution Detection through Relative Activation-Deactivation Abstractions}. In \bibinfo{booktitle}{\emph{2021 IEEE 32nd International Symposium on Software Reliability Engineering (ISSRE)}}. IEEE, \bibinfo{pages}{150--161}.
\newblock


\bibitem[Zhang et~al\mbox{.}(2020)]%
        {zhang2020clipped}
\bibfield{author}{\bibinfo{person}{Zhanyuan Zhang}, \bibinfo{person}{Benson Yuan}, \bibinfo{person}{Michael McCoyd}, {and} \bibinfo{person}{David Wagner}.} \bibinfo{year}{2020}\natexlab{}.
\newblock \showarticletitle{Clipped bagnet: Defending against sticker attacks with clipped bag-of-features}. In \bibinfo{booktitle}{\emph{2020 IEEE Security and Privacy Workshops (SPW)}}. IEEE, \bibinfo{pages}{55--61}.
\newblock


\bibitem[Zhong et~al\mbox{.}(2021)]%
        {zhong2021understanding}
\bibfield{author}{\bibinfo{person}{Ziyuan Zhong}, \bibinfo{person}{Yuchi Tian}, {and} \bibinfo{person}{Baishakhi Ray}.} \bibinfo{year}{2021}\natexlab{}.
\newblock \showarticletitle{Understanding local robustness of deep neural networks under natural variations}. In \bibinfo{booktitle}{\emph{Fundamental Approaches to Software Engineering: 24th International Conference, FASE 2021, Held as Part of the European Joint Conferences on Theory and Practice of Software, ETAPS 2021, Luxembourg City, Luxembourg, March 27--April 1, 2021, Proceedings 24}}. Springer International Publishing, \bibinfo{pages}{313--337}.
\newblock


\bibitem[Zhou et~al\mbox{.}(2023)]%
        {zhou2023majority}
\bibfield{author}{\bibinfo{person}{Qilin Zhou}, \bibinfo{person}{Zhengyuan Wei}, \bibinfo{person}{Haipeng Wang}, {and} \bibinfo{person}{WK Chan}.} \bibinfo{year}{2023}\natexlab{}.
\newblock \showarticletitle{A Majority Invariant Approach to Patch Robustness Certification for Deep Learning Models}. In \bibinfo{booktitle}{\emph{2023 38th IEEE/ACM International Conference on Automated Software Engineering (ASE)}}. IEEE, \bibinfo{pages}{1790--1794}.
\newblock


\end{thebibliography}

\end{document}